\documentclass[12pt]{article}
\usepackage{latexsym}
\usepackage{theorem}
\usepackage{graphicx}
\usepackage{mathtools}
\usepackage{xcolor}
\usepackage[page]{appendix}
\usepackage{soul}
\usepackage[breaklinks=true]{hyperref} 
\usepackage{enumitem}
\usepackage{url}

\theorembodyfont{\rmfamily}
\newtheorem{theorem}{Theorem}
\newtheorem{lemma}[theorem]{Lemma}
\newtheorem{proposition}[theorem]{Proposition}
\newtheorem{corollary}[theorem]{Corollary}
\newenvironment{proof}[1][Proof]{\noindent\textbf{#1.} }{\ \rule{0.5em}{0.5em}}

\theoremstyle{break}

\newtheorem{remark}[theorem]{Remark}

\newcommand{\blue}[1]{{\color{black}#1}}

\title{Queues with inspection cost: To see or not to see?}

\author{Jake Clarkson\thanks{jclarkson2302@gmail.com, corresponding author} \and Konstantin Avrachenkov\thanks{k.avrachenkov@inria.fr} \and Eitan Altman\thanks{eitan.altman@inria.fr}\\ 
}
\date{Centre Inria d'Université Côte d'Azur, Sophia Antipolis, 06902, France\thanks{The work was supported by the Project of Inria - Nokia Bell Labs
“Distributed Learning and Control for Network Analysis” and accepted by Springer journal Queueing Systems.}}

\begin{document}
\maketitle

\begin{abstract}
Consider an M/M/1-type queue where joining attains a known reward, but a known waiting cost is paid
per time unit spent queueing. In the 1960s, Naor showed that any arrival optimally joins
the queue if its length is less than a known threshold. Yet acquiring knowledge of the queue length often brings an additional cost, e.g., website loading time or data roaming charge.
Therefore, our model presents any arrival with three options: join blindly, balk blindly, or pay a known inspection cost to make the optimal joining decision by comparing the queue length to Naor's threshold. 
In a recent paper, Hassin and Roet-Green prove that a unique Nash equilibrium always exists and classify regions where the equilibrium probabilities are non-zero. We complement these findings with new closed-form expressions for the equilibrium probabilities in the majority of cases. Further, Hassin and Roet-Green show that minimizing inspection cost maximises social welfare. Envisaging a queue operator choosing where to invest, we compare the effects of lowering inspection cost and increasing the queue-joining reward on social welfare. We prove that the former dominates and that the latter can even have a detrimental effect on social welfare.
\end{abstract}

\section{Introduction and Model} \label{sec:intro}
In an increasingly digitalised and connected world, information is quickly becoming one of the most valuable assets. It is crucial, therefore, to be able to focus on helpful information and to use it efficiently. In this paper, following the work in \cite{Hassin2017}, we bring this question to an M/M/1 strategic queueing model, investigating how much customers value knowing the length of a queue that they can choose to join.

This issue is now widely relevant, as healthcare providers, voting locations, airport security and even amusement parks have started to publish live queue information. In \cite{Hassin2017}, several North American examples are provided; we notice that several UK organisations are following suit. For example, the British National Health Service (NHS) is currently developing an app which will provide wait times in emergency departments across the UK \cite{NHSquicker}. The theme park Alton Towers already has an app available where live queue times can be accessed \cite{AltonTowers}. In addition, both Edinburgh \cite{EdinburghAirport} and Manchester \cite{ManchesterAirport} airports now display the current security wait times on their homepages. 

With the abundance of smartphones, queue-length information is now accessible from anywhere for the majority of consumers. Despite being usually free of charge, its access often involves some cost or inconvenience to the customer. As noted by \cite{Hassin2017}, there may be a time cost whilst searching and waiting for a website or app to load. The consumer may also be required to divulge personal information before access, a process which is time-consuming, invades personal privacy, and may result in unwanted spam afterwards. In this paper, we additionally note that, if the consumer is away from home, accessing queue information may deplete battery, which could prevent smartphone use later in the day. Further, there may be a charge to use mobile data, particularly relevant if the customer is abroad where data roaming charges can be high. Considering customers as individuals who act in their own interest, we ask how costly must the access of queue-length information be for it to be no longer worth acquiring. 

To do this, we consider the extension made in \cite{Hassin2017} to the classic, continuous-time, single-server queueing model of Naor in \cite{Naor1969}. In Naor's model, customers arrive according to a Poisson process with rate $\lambda$, with each arrival deciding whether or not to join the queue. If they do not join, they leave the system; if they do join, they receive a known reward $R>0$ and incur a known waiting cost $c_w>0$ per unit time until they have been served. Only the customer at the front of the queue is served, with service times following an exponential distribution with mean $\mu$.

The aim of each customer is to minimise their own expected cost minus reward. The optimal joining strategy was derived in \cite{Naor1969} using the following logic. If the customer knows the length $q$ of the queue when they arrive (where $q$ includes any customer in service), then they can calculate the expected utility from joining to be
$$R-\frac{c_w(q+1)}{\mu},$$
optimally joining if and only if this value is positive; in other words, if joining is preferable to not joining. It follows that any customer will optimally join the queue if and only if the queue length is strictly less than
\begin{align*}
n_e &\equiv \max\left\{q: R- \frac{c_w q}{\mu}> 0\right\} \\ &= \max\left\{q: q < \frac{R\mu}{c_w}\right\},
\end{align*}
regardless of the joining strategies of other customers.

The novel extension of \cite{Hassin2017} allows a customer to learn the queue length on arrival if they pay a known inspection cost $C_I>0$.
If the customer does inspect the queue, afterwards they will enact Naor's optimal policy, joining if and only if the queue length is strictly less than $n_e$. If the customer chooses not to inspect, they can blindly either join or balk. In other words, there are three pure strategies: inspect, join and balk.

This paper makes several significant advancements to the ground-breaking analysis of \cite{Hassin2017} \blue{for the case where $\lambda < \mu$---the most common both in the queueing literature and the real world (since most queue managers ensure adequate resource to cover long-term demand).} Given the increasing availability of queue-length information in many real-life scenarios, this model is highly relevant, and hence any progress is of great practical value. We outline our contribution below.

A key result of \cite{Hassin2017} is to prove that a unique equilibrium always exists and to classify which of the three equilibrium probabilities (corresponding to the three pure strategies) are non-zero for different parameter values. 
Our work finds closed-form expressions for the equilibrium probabilities when at least one equilibrium probability is 0, which covers the vast majority of cases (all except the brown region in Figure \ref{fig:all}). 
Further, we explore new expressions and structural properties for the utility functions, which both greatly simplify the scenario where the reward $R$ is low and lead to a simpler alternative proof of a key step in the uniqueness proof of \cite{Hassin2017}.

\blue{In \cite{Hassin2017}, results are presented in terms of normalised results and costs.} \blue{In this paper, we use Figure \ref{fig:all} to} show visually the unique equilibrium as the \blue{unnormalised} inspection cost $C_I$ and reward $R$ vary, intuitively explaining several interesting patterns in behaviour. Such a depiction allows the service provider to both understand and visualise how marginal changes to \blue{the raw values of} $C_I$ and $R$ affect customer behaviour---a useful tool since $C_I$ and $R$ can be changed by altering business strategy. 
The reward $R$ represents the quality of service or product, which can be increased by investment and decreased under budget cuts. The inspection cost $C_I$ can be controlled by how easily queue-length information is available. As discussed in \cite{Hassin2017}, accessibility depends on the location of the information on the website, the design and loading time of the website, whether a specific smartphone app is dedicated to the queue length, and whether a registration process requiring personal details needs to be completed. To facilitate access even further, electronic signs could be installed around a hospital or amusement park with real-time updates to queue lengths.

Another key contribution of \cite{Hassin2017} is to analyse how marginal changes to $C_I$ affect several objectives. When revenue, defined as system throughput, is to be maximised, it is shown that the optimal inspection cost $C_I$ depends on the congestion, $\lambda/\mu$. When congestion is low, such as in a private hospital, $C_I$ should be set high, and when congestion is high, such as in a public hospital, $C_I$ should be set low. On the other hand, when social welfare, defined as the expected reward minus cost of an arrival at equilibrium, is to be maximised, it is shown, as expected, that $C_I$ should be taken as low as possible. 

We approach social welfare from a more practical angle. Any service provider has a limited budget. The amount of budget allocated to service quality will affect the reward $R$. 
In addition, the inspection cost $C_I$ is also affected by the budget's allocation, since the aforementioned factors which may decrease $C_I$ all come at a financial cost. Investment is needed to improve website design and loading times, or to build an app or signage specific to queue information.
Having the queue length on the front page of a website takes the place of other information or advertisement, creating an opportunity cost. Finally, potentially valuable customer data is forgone by removing or reducing any registration process. Therefore, we provide an analysis of how social welfare changes under marginal changes to $C_I$ and $R$ to inform a service provider how best to allocate their budget. We partition the parameter space into two: a region where the rate of social welfare improvement is greater decreasing $C_I$ than increasing $R$, and a region where the reverse holds. 
We also show that, due to an increase in queue lengths, increasing $R$ can sometimes paradoxically decrease social welfare.
 
Our analysis is clearly very relevant to a public service provider, where maximising social welfare is the main objective. For a more profit-driven private service, our work still applies, since a good customer experience is key to returning customers and long-term profit. 

To conclude the introductory section, we describe a few more relevant references in addition to \cite{Hassin2017} and \cite{Naor1969}. 
\blue{The work \cite{Hassin2020}, a follow up to \cite{Hassin2017}, analyzes how providing queue length information 
before the trip to the queue location affects customer behavior and system performance. The work \cite{Hassin1986} compares
the performance of observable and unobservable queues as congestion levels vary.}
Many strategic queueing models have been described and studied in the books \cite{Hassin2016book, HassinHaviv2003book,Stidham2009book}; \blue{see also the survey \cite{Ibrahim2018}, which focuses on the effect of information in queueing.} The amount of queue-length information available is crucial in the analysis of the strategic queueing models (see, e.g., \cite{Economou2021, Hassin2016book} for many references on the effect of information on strategic queueing). In some strategic queueing models, partial information can be provided to the users (see e.g., \cite{Altman2016} \blue{and \cite{Lingenbrink2019}}), or the users can be divided into groups with various levels of access to the information (e.g. \cite{Hu2018}). \blue{The work \cite{Burnetas2017} examines how customers behave strategically in a single-server queue where they cannot see the queue length upon arrival, but receive periodic updates about their position.} The recent paper \cite{Economou2022} from the Queueing System anniversary volume considers the question ``How much information should be given to the strategic customers of a queueing system?''. One way to calibrate the amount of information provided to the customers is to put a price upon it.

The layout of the paper is as follows. In Section \ref{sec:utility} we derive expressions for the utility functions for which we prove several key properties in Section \ref{sec:props}. In Section \ref{sec:find_eq}, we explore the structure of the unique equilibrium of \cite{Hassin2017} as $R$ and $C_I$ vary, adding expressions for the equilibrium probabilities and simplifying the case for low $R$. Section \ref{sec:SW} contains our study of social welfare as $R$ and $C_I$ change, to help a service provider best allocate an investment. Finally, Section \ref{sec:conclusion} concludes.

\section{Utility Functions} \label{sec:utility}
Recall that we consider a continuous-time, single-server queue with service rate $\mu$, to which customers arrive with rate $\lambda$. \blue{We assume that $\lambda < \mu$.} The queue length is not immediately observable to an arrival, who has three options:
\begin{enumerate}
\item Pay cost $C_I$ to inspect the queue length, then join the queue if and only if the queue length is strictly less than the optimal threshold $n_e$.
\item Join the queue without inspecting. 
\item Balk from the queue without inspecting. 
\end{enumerate} 
On joining the queue, the arrival receives a known reward $R>0$, but pays a known cost $c_w$ for each time unit until their service is completed. It was hence deduced in Section \ref{sec:intro} that the optimal threshold $n_e$ satisfies
\begin{equation} \label{eqn:ne_def}
n_e \equiv \max\left\{q: q < \frac{R\mu}{c_w}\right\}.
\end{equation}
As in \cite{Hassin2017}, to avoid trivialities, we will assume that $R \mu > c_w$ so that $n_e \geq 1$. 

\blue{The utility from balking is normalized to 0.} As in \cite{Hassin2017}, we write $U_I$ and $U_J$ for the expected utilities of inspecting and joining, respectively, but here express them as functions of the decision probabilities of all other customers, $P_I$, $P_J$ and $P_B$. Since $P_I+P_J+P_B=1$, it suffices to write $U_I(P_I,P_J)$ and $U_J(P_I,P_J)$. We will find alternative expressions for $U_I(P_I,P_J)$ and $U_J(P_I,P_J)$ to those of \cite{Hassin2017}, leading to several simplifications in special cases which will aid our analysis in later sections.

Write $Q$ for the random variable representing the queue length. The conditional expected utility satisfies
$$E[U|Q=q]=R-\frac{c_w(q+1)}{\mu};$$
from which it follows that we have 
\begin{align}
U_I(P_I,P_J)&=P(Q<n_e\mid P_I,P_J)\left(R-\frac{c_w\left(E\left[Q \mid Q<n_e,P_I,P_J\right]+1\right)}{\mu} \right) -C_I,  \label{eqn:VI}\\
U_J(P_I,P_J)&=R-\frac{c_w\left(E\left[Q\mid P_I,P_J\right]+1\right)}{\mu}. \label{eqn:VJ}
\end{align}

To obtain explicit expressions for \eqref{eqn:VI} and \eqref{eqn:VJ}, we first must derive the steady state distribution for a pair of decision probabilities $P_I$ and $P_J$. To facilitate the derivation, we write\footnote{Note that $\rho_L=\xi$ and $\rho_U=1-\eta$ in the notation of \cite{Hassin2017}.} $\rho_L=(P_I+P_J)\rho$ and $\rho_U=P_J\rho$, where \blue{$\rho=\lambda/\mu < 1$}. Therefore, $\rho_L$ is the traffic intensity when the queue is below $n_e$ (when only those balking do not join the queue) and $\rho_U$ is the traffic intensity when the queue is greater than or equal to $n_e$ (when only those who join without inspecting join the queue). 

Where $\pi_i$ is the steady state probability that the queue \blue{length is $i$}, as in \cite{Altman2016} and \cite{Hassin2017}, we derive
\begin{equation} \label{eqn:steadystate} 
\pi_i =
  \begin{cases}
    \rho_L^i \pi_0, & \text{for } i=0,1,\ldots, n_e-1 \\
    \rho_L^{n_e} \rho_U^{i-n_e} \pi_0, & \text{for } i=n_e,n_e+1,\ldots,
  \end{cases}
\end{equation}
by detailed balance, and since the steady state distribution sums to 1, we calculate
\begin{equation} \label{eqn:pi_0}
(\pi_0)^{-1}=\frac{1-\rho_L^{n_e}}{1-\rho_L}+\frac{\rho_L^{n_e}}{1-\rho_U}.
\end{equation}

To derive formulae for $U_I(P_I,P_J)$ and $U_J(P_I,P_J)$, we calculate the same conditional and unconditional expected queue lengths as \cite{Altman2016}. In particular, for $U_J(P_I,P_J)$, we use
$$E[Q|P_I,P_J]=\sum_{i=0}^\infty i \pi_i$$
in addition to \eqref{eqn:VJ}, \eqref{eqn:steadystate} and \eqref{eqn:pi_0} to obtain
\begin{equation} \label{eqn:VJ_form_gen}
U_J(P_I,P_J)= R-\frac{c_w(1-\rho_L+\rho_\Delta a+\rho_\Delta^2b)}{\mu(1-\rho_U)(1-\rho_U-\rho_\Delta \rho_L^{n_e})},
\end{equation}
\begin{equation} \label{eqn:EQ_gen_supp} 
\text{where} \quad \rho_\Delta=\rho_L-\rho_U, \; \; a=2-(n_e+2)\rho_L^{n_e}+n_e\rho_L^{n_e+1}, \; \; b=\frac{1-(n_e+1)\rho_L^{n_e}+n_e\rho_L^{n_e+1}}{1-\rho_L}. 
\end{equation}
To calculate $U_I(p_I,p_J)$, we use
$$P[Q<n_e|P_I,P_J]=\sum_{i=0}^{n_e-1} \pi_i \quad \text{and} \quad E[Q|Q<n_e,P_I,P_J]=\frac{1}{P[Q<n_e|P_I,P_J]}\sum_{i=0}^{n_e-1} i \pi_i,$$
in addition to \eqref{eqn:VI}, \eqref{eqn:steadystate} and \eqref{eqn:pi_0}, to obtain
\begin{equation} \label{eqn:VI_form_gen}
U_I(P_I,P_J)=\frac{(1-\rho_U)}{(1-\rho_U-\rho_\Delta \rho_L^{n_e})}\left[R(1-\rho_L^{n_e})-\frac{c_w(1+n_e\rho_L^{n_e+1}-(n_e+1)\rho_L^{n_e})}{\mu(1-\rho_L)}\right]-C_I.
\end{equation}

We next look at the special cases where no or all customers inspect the queue.

\paragraph{No Customers Inspect the Queue}
In this case $P_I=0$, which leads to an \blue{unobservable} M/M/1 queue with arrival rate $\lambda P_J$ for all states. Therefore, $\rho_L=\rho_U=\rho P_J$ and hence $\rho_\Delta=0$, so \eqref{eqn:VJ_form_gen} simplifies to give
\begin{equation} \label{eqn:VJ_0,p_J}
U_J(0,P_J)=R-\frac{c_w}{\mu(1-\rho P_J)}=R-\frac{c_w}{\mu-\lambda P_J},
\end{equation}
which is consistent with the known expected queue length for an \blue{unobservable} M/M/1 queue.
Additionally, \eqref{eqn:VI_form_gen} simplifies to
\begin{equation} \label{eqn:VI_0,p_J}
U_I(0,P_J)=R(1-(\rho P_J)^{n_e})-\frac{c_w(1+n_e(\rho P_J)^{n_e+1}-(n_e+1)(\rho P_J)^{n_e})}{\mu-\lambda P_J}-C_I.
\end{equation}

\paragraph{All Customers Inspect the Queue}
In this case $P_I=1$, which leads to an \blue{observable} M/M/1 queue where all arriving customers join the queue if and only if the queue is strictly less than length $n_e$; this policy was studied in \cite{Naor1969}. We have $\rho_U=0$ and $\rho_L=\rho_\Delta=\rho$, so \eqref{eqn:VJ_form_gen} simplifies to give
\begin{equation} \label{eqn:VJ_1,*}
U_J(1,0)=R-\frac{c_w(1-\rho^{n_e+1}-(1-\rho)(n_e+1)\rho^{n_e+1})}{(\mu-\lambda)(1-\rho^{n_e+1})},
\end{equation}
which agrees with the calculations of \cite{Naor1969}.
Additionally, \eqref{eqn:VI_form_gen} simplifies to
\begin{equation} \label{eqn:VI_1,*}
U_I(1,0)=\frac{R(1-\rho^{n_e})}{1-\rho^{n_e+1}}-\frac{c_w(1+n_e\rho^{n_e+1}-(n_e+1)\rho^{n_e})}{(\mu-\lambda)(1-\rho^{n_e+1})}-C_I.
\end{equation}

We finish this section by deriving an expression for $U_{J-I}(P_I,P_J) \equiv U_J(P_I,P_J)-U_I(P_I,P_J)$ independently of \eqref{eqn:VJ_form_gen} and \eqref{eqn:VI_form_gen}. Consider two arrivals to a queue in steady state when all other customers play $(P_I,P_J)$. The first arrival joins the queue blindly, whilst the second inspects. With probability $P(Q<n_e \mid P_I,P_J)$ both arrivals join the queue, with the inspecting arrival incurring an extra cost $C_I$. With probability $P(Q\geq n_e \mid P_I,P_J)$, the inspecting arrival does not join and pays only cost $C_I$, but the blind arrival does join, incurring a net reward of
$$R-\frac{c_w(E[Q|Q\geq n_e,P_I,P_J]+1)}{\mu}.$$
It follows that
\begin{align} \label{eqn:V_J-I_eqn}
U_{J-I}(P_I,P_J)=P(Q\geq n_e|P_I,P_J)\left[R-\frac{c_w(E[Q|Q\geq n_e,P_I,P_J]+1)}{\mu}\right]+C_I.
\end{align}
Conditional on being greater than length $n_e$, we have a standard M/M/1 queue with traffic intensity $\rho_U \equiv P_J \rho$, so the conditional expectation is simply $E[Q|Q\geq n_e,P_I,P_J]=n_e+\rho_U/(1-\rho_U)$. Using the steady state distribution to evaluate $P(Q\geq n_e|P_I,P_J)$, we hence have
\begin{align} \label{eqn:V_J-I_eqn_closed}
U_{J-I}(P_I,P_J)=\frac{\rho_L^{n_e}(1-\rho_L)}{1-\rho_U+(\rho_U-\rho_L)\rho_L^{n_e}}\left[R-\frac{c_w(n_e+1+\blue{\rho_U/(1-\rho_U)})}{\mu}\right]+C_I.
\end{align}
Both \eqref{eqn:V_J-I_eqn} and \eqref{eqn:V_J-I_eqn_closed} will be useful in later sections.

\section{Properties of Utility Functions} \label{sec:props}
In this section, we show several monotonicity properties of the utilities $U_J$, $U_I$ and their difference $U_{J-I}$.
First, we need a lemma, which shows that increasing traffic intensity leads to a stochastically larger queue and an increase in the value of several properties.
\begin{lemma} \label{lem:monoPandE}
Suppose both $\rho_L$ and $\rho_U$ are increased, at least one by a non-zero amount. Then we have the following \blue{effects}:
\begin{enumerate}[label=(\roman*)]
\item \blue{For any fixed $x \in \{0,1,2,\ldots\}$, the value of $P(Q\geq x)$ will strictly increase.} \label{cond:lemP(Q>=x)}
\item $E[Q]$ \blue{will} strictly increase. \label{cond:lemE[Q]}
\item $E[Q\mid Q<n_e]$ \blue{will} increase.\label{cond:lemE[Q|Q<l]}
\item $E[Q\mid Q\geq n_e]$ \blue{will} increase. \label{cond:lemE[Q|Q>=l]}
\end{enumerate}
\end{lemma}
\begin{proof}
Write $\tilde{Q}$ for the random variable representing the queue length and $\tilde{\pi}_i$, $i=0,1,\ldots$, for the steady state probabilities when the traffic intensities are $\tilde{\rho}_L\geq \rho_L $ and $\tilde{\rho}_U\geq \rho_U $. As in the lemma statement, assume that at least one of these inequalities is strict.

By \eqref{eqn:steadystate}, we have
\begin{equation*}
\frac{\tilde{\pi}_i}{\pi_i} =
  \begin{dcases}
    \frac{\tilde{\pi}_0}{\pi_0} \left(\frac{\tilde{\rho}_L}{\rho_L}\right)^i & \text{for } i=0,1,\ldots, n_e-1 \\
    \frac{\tilde{\pi}_0}{\pi_0} \left(\frac{\tilde{\rho}_L}{\rho_L}\right)^{n_e} \left(\frac{\tilde{\rho}_U}{\rho_U}\right)^{i-n_e} & \text{for } i=n_e,n_e+1,\ldots
  \end{dcases}
\end{equation*}
Consider the sequence
\begin{equation} \label{eqn:proof_seq}
\left\{\frac{\tilde{\pi}_i}{\pi_i}: i=0,1,2,\ldots \right\}.
\end{equation}
Since $\rho_L \leq \tilde{\rho}_L$ and $\rho_U \leq \tilde{\rho}_U$, the sequence must be increasing, and, since at least one of the inequalities is strict, it cannot be constant. \blue{Therefore, $\tilde{Q}$ is strictly greater than $Q$ in the likelihood ratio order and hence, by Theorem 1.C.1 of \cite{shaked2007stochastic}, also in the usual stochastic order\footnote{Theorem 1.C.1 shows non-strict dominance in likelihood ratio order implies non-strict dominance in the usual stochastic order. Yet, the proof may be adapted to show a similar result for strict dominance holds. Note that, to prove parts (iii) and (iv) of the lemma, we do not need a strict version of Theorem 1.C.1.}. Parts \ref{cond:lemP(Q>=x)} and \ref{cond:lemE[Q]} immediately follow.

To prove parts \ref{cond:lemE[Q|Q<l]} and \ref{cond:lemE[Q|Q>=l]}, we combine Theorems 1.C.6 and 1.C.1 of \cite{shaked2007stochastic}, which show that if $\tilde{Q}$ is greater than $Q$ in the likelihood ratio order, then $[\tilde{Q} | \tilde{Q} \in A]$ is greater than $[Q | Q \in A]$ in the usual stochastic order for any measurable set $A$.}
\end{proof}

Now we consider the expected utility functions $U_J(P_I,P_J)$ and $U_I(P_I,P_J)$. Either when blindly joining the queue ($U_J$) or inspecting it first ($U_I$), it is logical that expected utility decreases as queue congestion increases, but how do the actions of previous arrivals, which are determined by $P_I$ and $P_J$, affect congestion? Recall that any previous arrival had three options: inspect, join blindly or balk. To lower congestion, we would clearly like to switch either a queue inspector or blind joiner for a balker, since the latter never joins the queue. Further, since a blind queue joiner \emph{always} joins the queue whilst a queue inspector only joins \emph{if} the current queue length is below $n_e$, we would also like to swap a blind joiner for a queue inspector. We combine Lemma \ref{lem:monoPandE} with the formulae in \eqref{eqn:VI} and \eqref{eqn:VJ} to formally state and prove these assertions below.

\begin{lemma} \label{lem:mono_all_U}
The following is true for $U=U_J$ and $U=U_I$. For any strategy pair $(P_I,P_J)$ and for any $x>0$ such that the arguments of $U$ remain a valid strategy pair, we have
\begin{enumerate}[label={(\alph*)}]
\item $U(P_I,P_J) > U(P_I,P_J+x)$ \, \qquad     (balkers switched with blind joiners) \label{lempart:bk-bj}
\item $U(P_I,P_J) > U(P_I+x,P_J)$   \, \qquad   (balkers switched with queue inspectors) \label{lempart:bk-qi}
\item $U(P_I,P_J) > U(P_I-x,P_J+x)$ \; (queue inspectors switched with blind joiners) \label{lempart:qi-bj}
\end{enumerate}

\end{lemma}
\begin{proof}
We begin by showing that Lemma \ref{lem:monoPandE} applies in \ref{lempart:bk-bj}, \ref{lempart:bk-qi} and \ref{lempart:qi-bj}. Recall that the traffic intensities below and above the threshold are given by $\rho_L\equiv (P_I+P_J)\rho$ and $\rho_U \equiv P_J \rho$. In \ref{lempart:bk-bj}, $P_J$ is increased and $P_I$ fixed, so both $\rho_L$ and $\rho_U$ increase and Lemma \ref{lem:monoPandE} applies. In \ref{lempart:bk-qi}, $P_I$ is increased and $P_J$ fixed. Therefore, $\rho_L$ increases and $\rho_U$ stays the same, which again satisfies the conditions of Lemma \ref{lem:monoPandE}. Finally, in \ref{lempart:qi-bj} the sum $P_I+P_J$ is fixed whilst $P_J$ increases. This leads to a fixed $\rho_L$ but an increase in $\rho_U$, which again permits Lemma \ref{lem:monoPandE}.

To complete the proof, we show that whenever Lemma \ref{lem:monoPandE} applies both $U_I$ and $U_J$ strictly decrease. For $U_J$, this is immediate from \ref{cond:lemE[Q]} of Lemma \ref{lem:monoPandE} and inspection of \eqref{eqn:VJ}. For $U_I$, examine the formula in \eqref{eqn:VI}. Note that we have
\begin{equation} \label{eqn:UI_term_strpos}
R-\frac{c_w\left(E\left[Q \mid Q<n_e,P_I,P_J\right]+1\right)}{\mu}\geq R-\frac{c_w n_e}{\mu}>0,
\end{equation}
with the last inequality by the definition of $n_e$. Therefore, the left-hand term in \eqref{eqn:VI} is \blue{the product of two positive-valued expressions. When Lemma \ref{lem:monoPandE} applies, the first expression, $P(Q<n_e|P_I,P_J)$, strictly decreases by \ref{cond:lemP(Q>=x)}, and the second expression, the left-hand term in \eqref{eqn:UI_term_strpos}, decreases by \ref{cond:lemE[Q|Q<l]} (which shows that $E[Q|Q<n_e,P_I,P_J]$ increases).
The product of the two expressions therefore strictly decreases when Lemma \ref{lem:monoPandE} applies.}
Since the right-hand term in \eqref{eqn:VI}, namely $-C_I$, is fixed, the result for $U_I$ follows. 
\end{proof}

Lemma \ref{lem:mono_all_U} shows that certain changes in the strategies of previous arrivals which increase queue congestion lead to a decrease in the expected utility both from joining blindly and inspecting. In other words, as queue congestion increases, the value in balking over both inspecting or blindly joining increases. Now we ask: how does the value in blindly joining over inspecting behave as queue congestion increases? One would expect it to decrease, since blindly joining commits to joining a queue of any length whilst inspecting permits balking if the queue length is seen to be too long. We may again use Lemma \ref{lem:monoPandE} to formalise this result below.

\begin{corollary} \label{corol:mono_diff}
The three properties \ref{lempart:bk-bj}, \ref{lempart:bk-qi} and \ref{lempart:qi-bj} in Lemma \ref{lem:mono_all_U} all hold for $U=U_{J-I}$ where $U_{J-I}(P_I,P_J) \equiv U_J(P_I,P_J)-U_I(P_I,P_J)$.
\end{corollary}
\begin{proof}
Recall that the proof of Lemma \ref{lem:mono_all_U} showed that Lemma \ref{lem:monoPandE} applies under \ref{lempart:bk-bj}, \ref{lempart:bk-qi} and \ref{lempart:qi-bj}. We will complete the proof by showing that the expression in \eqref{eqn:V_J-I_eqn} strictly decreases whenever Lemma \ref{lem:monoPandE} applies. The argument is similar to the $U=U_I$ case in the proof of Lemma \ref{lem:mono_all_U}. Note that we have
\begin{equation} \label{eqn:UJ-I_term_strpos}
R-\frac{c_w\left(E\left[Q \mid Q \geq n_e,P_I,P_J\right]+1\right)}{\mu}\leq R-\frac{c_w (n_e+1)}{\mu}<0,
\end{equation}
with the last inequality by the definition of $n_e$. Therefore, the left-hand term in \eqref{eqn:V_J-I_eqn} is \blue{the product of one positive and one negative expression. When Lemma \ref{lem:monoPandE} applies, the positive expression, $P(Q\geq n_e|P_I,P_J)$, strictly increases by \ref{cond:lemP(Q>=x)}, and the negative expression, the left-hand term in \eqref{eqn:UJ-I_term_strpos}, decreases by \ref{cond:lemE[Q|Q<l]} (which shows that $E[Q|Q\geq n_e,P_I,P_J]$ increases).
It follows that the product of the two expressions strictly decreases when Lemma \ref{lem:monoPandE} applies.}
Since the right-hand term in \eqref{eqn:V_J-I_eqn}, $C_I$, is fixed, the result follows. 
\end{proof}

\begin{remark} \label{remark:compare_Hassin_mono}
Corollary \ref{corol:mono_diff} is stated geometrically in \cite{Hassin2017} as \emph{Lemma 2}, a result which is a key step in the proof of their \emph{Theorem 1}, that a unique equilibrium exists. We believe that our Corollary \ref{corol:mono_diff} offers a simpler and more intuitive proof of this result than that in \cite{Hassin2017}.
\end{remark}

\section{Equilibrium Structure as $R$ and $C_I$ Vary} \label{sec:find_eq}
Recall that a mixed strategy for an arrival is determined by a set of non-negative probabilities $P_I,P_J,P_B$ satisfying $P_I+P_J+P_B=1$. Therefore, we may denote a mixed strategy by the 2-tuple $(P_I,P_J)$.
As in \cite{Hassin2017}, we consider (Nash) equilibria defined as follows. A mixed strategy $(P_I,P_J)$ is an \emph{equilibrium} if and only if $(P_I,P_J)$ is a best response when all other customers play $(P_I,P_J)$. The authors of \cite{Hassin2017} prove in their \emph{Theorem 1} that a unique equilibrium of the type described above always exists. In this section, our aim is to adapt and extend the results of \cite{Hassin2017} \blue{for the case $\rho < 1$} to depict this unique equilibrium as $R$ and $C_I$ vary and explain the resulting patterns. We will also derive closed-form expressions for the equilibrium probabilities in the cases where at least one of $P_I, P_J$ and $P_B$ is 0, which, as will be shown in Figure \ref{fig:all}, cover the majority of cases.

\blue{For the case $\rho< 1$,} the authors of \cite{Hassin2017} define three scenarios in their \emph{Proposition 1}, which are determined by the value of $\nu\equiv R\mu/c_w$. Each scenario sees different behaviour in the equilibrium probabilities as $C_I$ increases. These scenarios will also be key to our analysis, so we rewrite them in terms of $R$ below.
\begin{align*}
\text{Scenario 1:} &\quad R>\frac{c_w}{\mu-\lambda}\\
\text{Scenario 2:} &\quad c_w\left(\frac{2}{\mu}-\frac{1}{\mu+\lambda}\right)<R\leq \frac{c_w}{\mu-\lambda} \\ 
\text{Scenario 3:} &\quad \frac{c_w}{\mu} < R \leq c_w\left(\frac{2}{\mu}-\frac{1}{\mu+\lambda}\right)
\end{align*}

We begin our analysis with two key properties of scenarios 1 and 3, respectively.
\begin{lemma} \label{lem:Scen1_prop}
We have $U_J(P_I,P_J)> 0$ for every $P_I,P_J$ if and only if we are in Scenario 1.  
\end{lemma}
\begin{proof}
For any $P_I,P_J \in [0,1]$, we have
\begin{equation} \label{eqn:proof_2<0_lrg_p}
U_J(P_I,P_J)\geq U_J(1-P_J,P_J) \geq U_J(0,1)=R-\frac{c_w}{\mu-\lambda}.
\end{equation}
\blue{The first inequality holds by taking $x=1-P_J-P_I$ in part \ref{lempart:bk-qi} of Lemma \ref{lem:mono_all_U}. The second inequality holds by taking $x=1-P_J$ in part \ref{lempart:qi-bj} of Lemma \ref{lem:mono_all_U}. The final equality, which links the terms in \eqref{eqn:proof_2<0_lrg_p} to the definition of Scenario 1, holds by \eqref{eqn:VJ_0,p_J}. We now use \eqref{eqn:proof_2<0_lrg_p} to prove the lemma.} By its definition, if we are in scenario 1 then all terms in \eqref{eqn:proof_2<0_lrg_p} are positive, proving the backwards implication. If we are not in scenario 1, then $U_J(0,1) \leq 0$, proving the (contrapositive of the) forwards implication.
\end{proof}

Lemma \ref{lem:Scen1_prop} shows that no matter the strategies of other customers, if a new arrival chooses not to inspect the queue, it is always better to join than to balk in scenario 1. Any equilibrium in scenario 1 therefore has $P_B=0$.

\begin{proposition} \label{lem:Scen3_prop}
We have $U_J(1,0)\leq 0$ if and only if we are in Scenario 3. Further, both can only occur if $n_e=1$. 
\end{proposition}
\begin{proof}
From \eqref{eqn:VJ_1,*}, we see that $U_J(1,0)\leq 0$ if and only if
\begin{align}
R&\leq\frac{c_w(1-\rho^{n_e+1}-(1-\rho)(n_e+1)\rho^{n_e+1})}{(\mu-\lambda)(1-\rho^{n_e+1})} \nonumber \\
&=\frac{c_w}{\mu-\lambda} \left[1-\frac{(1-\rho)(n_e+1)\rho^{n_e+1}}{1-\rho^{n_e+1}}\right] \label{eqn:rhs_cons_proof_UJ(1,0)}
\end{align}
Suppose that $n_e=1$, so $R\in[c_w/\mu,2c_w/\mu]$. With $n_e=1$, the right-hand side of \eqref{eqn:rhs_cons_proof_UJ(1,0)} simplifies to
\begin{align*}
\frac{c_w}{\mu-\lambda} \left[\frac{(1-\rho)(2\rho+1)}{1+\rho}\right]=\frac{c_w}{\mu} \left[\frac{2\rho+1}{1+\rho}\right]=\frac{c_w}{\mu} \left[2-\frac{1}{\rho+1}\right]=c_w \left[\frac{2}{\mu}-\frac{1}{\mu+\lambda}\right],
\end{align*}
which is precisely the condition of Scenario 3. Therefore, when $n_e=1$, Scenario 3 occurs if and only if $U_J(1,0)\leq 0$.

Now we consider $n_e>1$, in which case $R>2c_w/\mu$. Since $1/(\mu+\lambda)>0$, by definition, Scenario 3 cannot occur when $n_e>1$. To complete the proof, we show that the same is true for $U_J(1,0)\leq 0$.

Write $R=(n_ec_w+k)/\mu$. If we then solve \eqref{eqn:rhs_cons_proof_UJ(1,0)} as an equality for $k$ we obtain
\begin{align}
k&=c_w\left[\frac{1}{1-\rho}-\frac{(n_e+1)\rho^{n_e+1}}{1-\rho^{n_e+1}}-n_e\right] \nonumber \\
&=c_w\left[\frac{\sum_{i=0}^{n_e} \rho^i-(n_e+1)\rho^{n_e+1}}{(1-\rho)\sum_{i=0}^{n_e} \rho^i}-n_e\right], \label{eqn:k_express_proof}
\end{align}
where the equality follows from the identity $1-\rho^{n_e+1}=(1-\rho) \sum_{i=0}^{n_e} \rho^i$. Note that the top of the fraction in \eqref{eqn:k_express_proof} has $(1-\rho)$ as a factor, so we may simplify to
\begin{equation} \label{eqn:k_express_proof_2}
k=c_w\left[\frac{\sum_{i=0}^{n_e} (i+1)\rho^i}{\sum_{i=0}^{n_e} \rho^i}-n_e\right] = c_w\left[\frac{\sum_{i=0}^{n_e}(i+1-n_e) \rho^i}{\sum_{i=0}^{n_e} \rho^i}\right].
\end{equation}
Note that only the last term in the sum of the numerator of \eqref{eqn:k_express_proof_2} is positive. Therefore, considering only the first and last terms yields the inequality
$$k\leq c_w\left[\frac{1-n_e+\rho^{n_e}}{\sum_{i=0}^{n_e} \rho^i}\right]<0$$
when $n_e>1$.
The condition for $U_J(1,0)\leq 0$ therefore takes the form
\begin{equation} \label{eqn:k_proof_eqn}
R\leq \frac{n_ec_w+k}{\mu}<\frac{c_w n_e}{\mu}.
\end{equation}
Yet, when the threshold is $n_e$, we have $R\geq n_ec_w/\mu$, so \eqref{eqn:k_proof_eqn} never holds when $n_e>1$, completing the proof.
\end{proof}

Proposition \ref{lem:Scen3_prop} makes a considerable advancement on \blue{the analysis of \cite{Hassin2017} for $\rho<1$} by showing that Scenario 3 can \emph{only} occur if $n_e=1$, in other words, if inspectors will only join the queue if it is empty. The result also says that, in scenario 3, if all other customers inspect the queue then an arrival who chooses not to inspect always does better balking than joining, with the opposite true in Scenarios 1 and 2. This observation affects whether it is joiners or balkers who break the equilibrium $P_I=1$ as $C_I$ increases. 

In \emph{Proposition 1} of \cite{Hassin2017}, regions with different types of equilibrium are defined by the value of $\kappa\equiv C_I\mu/c_w$. To create our visualisation in terms of $R$ and $C_I$, we adapt relevant parts of \emph{Proposition 1} to define regions by the value of $C_I$. These regions will be determined by inequalities, one side of which will be $C_I$, and the other a function of the reward $R$ denoted $C_*^*(R)$. Throughout, we shall use the following naming procedure for these functions. The superscript will denote the scenarios on which the inequality applies. The subscript will denote both the identity (I, J or B) and the value (0 or 1) of the equilibrium probability the inequality concerns. For example, the inequality $C_{I,0}^{2,3}$ will define the region where the equilibrium has $P_I=0$ in scenarios 2 and 3. For each inequality, we provide a reference in a footnote to the equivalent inequality of \cite{Hassin2017}.

We split our analysis into three subsections depending on the value of $P_I$ at the equilibrium: no customers inspect ($P_I=0$), all customers inspect ($P_I=1$) and some customers inspect ($P_I\in(0,1)$). \blue{For each, we discuss the conditions on $R$ and $C_I$ under which each equilibrium type arises.}

\subsection{\blue{Conditions Where} No Customers Inspect} \label{sec:no_custs}
\blue{When no customers inspect ($P_I=0$), our model simplifies to the unobservable M/M/1 queue of \cite{Edelson1975}.} Table \ref{tab:HR_no_custs} adapts the inequalities of \emph{Proposition 1} in \cite{Hassin2017} \blue{to show the conditions under which the equilibrium has} $P_I=0$ across the three scenarios\footnote{The inequalities $\kappa \geq K_5$ and $\kappa\geq K_4$ from page 811 of \cite{Hassin2017} are (respectively) equivalent to the inequalities $C_I\geq C_{I,0}^{1}(R)$ and $C_I\geq C_{I,0}^{2,3}(R)$ from Table \ref{tab:HR_no_custs}.}.
To complement these findings, \blue{by setting our expression for $U_J(0,P_J)$ in \eqref{eqn:VJ_0,p_J} equal to 0 (which equates the utility from blindly joining and balking), we find that \blue{the equilibrium probability of blindly joining the queue in scenarios 2 and 3 satisfies}
\begin{equation*}
\tilde{P}_J=\dfrac{R\mu-c_w}{R\lambda}.
\end{equation*}}
Note that, as $R$ increases, joining blindly becomes more attractive, so we see $\tilde{P}_J$ increasing as a function of $R$ as a result. 

\begin{table}[h]
\caption{Summary of equilibria concerning $P_I=0$.} \label{tab:HR_no_custs}
\centering
\begin{tabular}{|c|c|c|c|}
\hline
 Scenario & Range of $R$ & Condition on $C_I$ & Equilibrium \\
 \hline & & & \\ [-2.5ex]
 1 & $R > \dfrac{c_w}{\mu-\lambda} $ & $C_I\geq C_{I,0}^{1}(R)\equiv \rho^{n_e} \left[ \dfrac{c_w n_e}{\mu}+\dfrac{c_w}{\mu-\lambda}-R\right]$ & $(0,1)$ \\[1.7ex]
 2 \& 3 & $R \in \left(\dfrac{c_w}{\mu},\dfrac{c_w}{\mu-\lambda}\right]$ & $C_I\geq C_{I,0}^{2,3}(R) \equiv \dfrac{n_ec_w}{\mu}\left(1-\dfrac{c_w}{R\mu} \right)^{n_e}$ & $(0,\tilde{P}_J)$\\[1.7ex]
 \hline
\end{tabular}
\end{table}

\begin{figure} 
\centering
\includegraphics[width=1\textwidth]{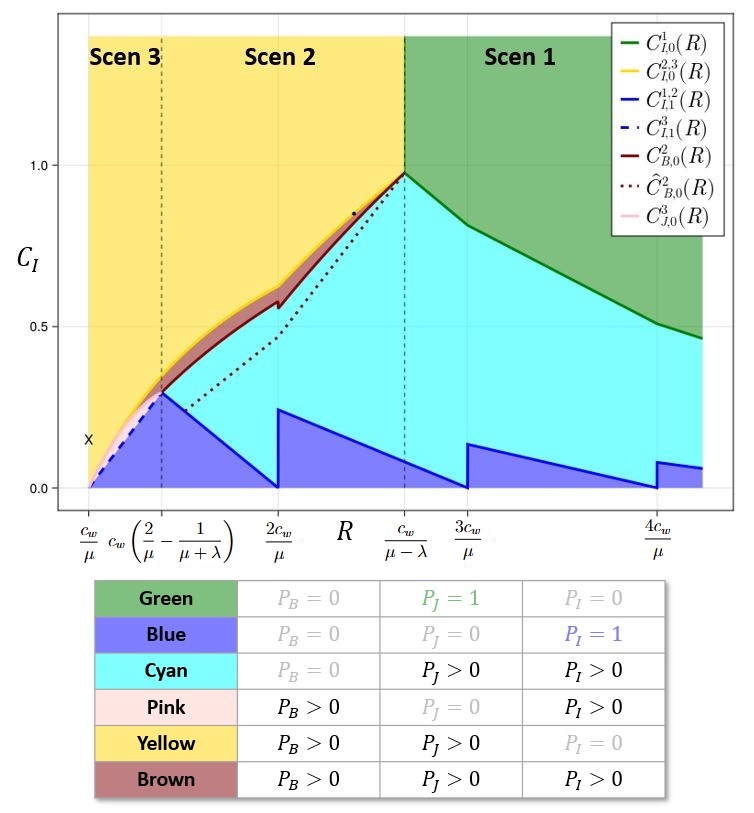} 
\caption{(Best viewed in colour) For $\lambda=0.5, \mu=0.8$ and $c_w=1$, equilibria as $R$ and $C_I$ vary. The scenarios 1, 2 and 3 are labelled at the top of the plot and separated by the black, vertical, dashed lines. The increments of $c_w/\mu$ on the $R$-axis demonstrate the regions where $n_e$ is constant. The colour represents a different equilibrium form in that region. The lines represent various functions from Section \ref{sec:find_eq} which define the coloured regions.} \label{fig:all} 
\end{figure}

\blue{In the remainder of this subsection, to understand customers' motivation to pay the inspection cost for different rewards $R$, we will investigate the functions $C_{I,0}^1(R)$ and $C_{I,0}^{2,3}(R)$, which bound the region where $P_I=0$ at equilibrium across all three scenarios. The properties below, which are visible in Figure \ref{fig:all}, will be formally proven in Appendix \ref{append:proofs_no_custs}.

\paragraph{Properties of $C_{I,0}^1(R)$:} For any $R>c_w/\mu$, the function $C_{I,0}^1(R)$ is piecewise linear with breakpoints at $xc_w/\mu$ for $x \in \{2,3,\ldots\}$. Further, $C_{I,0}^1(R)$ is continuous, strictly decreasing, strictly positive, and has limit 0 as $R\rightarrow \infty$.

\paragraph{Properties of $C_{I,0}^{2,3}(R)$:}
For any $R > c_w/\mu$, the function $C_{I,0}^{2,3}(R)$ is piecewise differentiable with breakpoints at $xc_w/\mu$ for $x \in \{2,3,\ldots\}$. Further, $C_{I,0}^{2,3}(R)$ is continuous, strictly increasing and strictly positive on this range.}

\vspace{15pt}

\blue{The breakpoints of both functions occur when the optimal threshold $n_e$ increases from $x-1$ to $x$---the precise moment at which joining a queue of length $x$ becomes the optimal choice.} We note that the strictly increasing property of $C_{I,0}^{2,3}(R)$ is in direct contrast to the strictly decreasing $C_{I,0}^{1}(R)$, even though both define a region on which the equilibrium has $P_I=0$, just in different scenarios. We explain this difference in the following.

The benefit of inspecting the queue is to compare the learned queue-length to the optimal threshold $n_e$ to ensure making the optimal joining decision. By Lemma \ref{lem:Scen1_prop}, if the queue is not inspected in scenario 1, it is always better (in expectation) to blindly join the queue than balk. Therefore, the equilibrium when $P_I=0$ is $P_J=1$ (). As the reward $R$ received on joining the queue increases, our confidence in joining blindly grows and grows, so paying a cost $C_I$ to inspect and ensure we make the optimal decision becomes less and less attractive. Therefore, $C_{I,0}^1(R)$, the maximum price we are willing to pay for queue-length information if everyone else joins the queue, decreases with $R$, decaying to 0 as we approach the $n_e=\infty$ case where queue-length information is useless.

On the other hand, in scenarios 2 and 3, the equilibrium when $P_I=0$ has $P_J=\tilde{P}_J$ satisfying $U_J(0,\tilde{P}_J)=0$; in other words, regardless of $R$, any arrival has utility 0 attained either from balking or joining. When $R=c_w/\mu$, inspecting the queue length is useless since $n_e=0$ and the queue is never optimally joined. As $R$ increases, however, it becomes more attractive to join the queue for any given queue length. Therefore, the expected utility of 0 attained at the equilibrium $(0,\tilde{P}_J)$ becomes less and less attractive compared to the expected utility that could be achieved by inspecting the queue with some strictly positive probability. 
Therefore, we see $C_{I,0}^{2,3}(R)$ increasing up to $R=c_w/(\mu-\lambda)$, when $\tilde{P}_J$ reaches 1 and we enter scenario 1.

Figure \ref{fig:all} shows that $C_{I,0}^{2,3}(R)$ and $C_{I,0}^{1}(R)$ coincide at $R=c_w/(\mu-\lambda)$, a fact which is easily proven to always be true by substitution into the definitions in Table \ref{tab:HR_no_custs}. 
Therefore, for $R>c_w/\mu$ there exists a \emph{continuous} boundary separating the area where the equilibrium involves customers never inspecting the queue length. We summarise in the following.
\begin{proposition} \label{thm:no_toll}
Write 
\begin{equation*}
C_{I,0}(R) \equiv \begin{cases}
C_{I,0}^{2,3}(R) \quad &\text{for }\frac{c_w}{\mu} < R \leq \frac{c_w}{\mu-\lambda},\\
C_{I,0}^{1}(R) \quad &\text{for }R \geq \frac{c_w}{\mu-\lambda}.
\end{cases}
\end{equation*}
Then $C_{I,0}(R)$ is continuous and piecewise \blue{differentiable} with breakpoints at $c_w/(\mu-\lambda)$ and $xc_w/\mu$ for $x\in \{2,3,\ldots\}$. Whenever $C_I \geq C_{I,0}(R)$, the unique equilibrium involves no customer inspecting the queue. 
\end{proposition} 

\subsection{\blue{Conditions Where} All Customers Inspect} \label{sec:all_custs}
\blue{When all customers inspect ($P_I=1$), our model simplifies to the observable M/M/1 queue of \cite{Naor1969}. Clearly this is the case} if $C_I=0$, but as inspecting becomes more and more costly, the equilibrium probability $P_I$ eventually begins to fall from 1. We find that the value of $C_I$ at which this occurs is given by a different formula\footnote{The work \cite{Hassin2017} derives the two formulas in \emph{Proposition 1}, but does not note they are split by scenarios.} for scenario 3 than scenarios 1 and 2. The reason is Proposition \ref{lem:Scen3_prop}: in scenario 3, the $P_I=1$ equilibrium is broken when balking becomes the best response to $P_I=1$, whilst in scenarios 1 and 2 the change occurs when blindly joining becomes the best response.

For scenario 3, \cite{Hassin2017} presents an inequality\footnote{Denoted $\kappa \leq K_1 \equiv \min (S_1,S_2)$ on page 811 of \cite{Hassin2017}. In scenario 3 when $\rho<1$, we find that $K_1=S_2$. Note the term $n_e\rho^{n_e-1}$ in the numerator of the right-hand term of $S_2$ should be $n_e\rho^{n_e+1}$.} for the point at which $P_I=1$ ceases to be the equilibrium which involves the ratios of polynomials of degree $n_e$ and $n_e+1$.
For the case $\rho < 1$ \blue{studied in this paper,} we use Proposition \ref{lem:Scen3_prop}, which shows that scenario 3 is only possible when $n_e=1$, to greatly simplify this inequality. Our new inequality, obtained by setting $n_e=1$ in that of \cite{Hassin2017}, is given as $C_I \leq C_{I,1}^3(R)$ in Table \ref{tab:HR_all_custs}, which also contains an adapted version of an inequality\footnote{As in previous footnote, but in scenarios 1 and 2, we find that $K_1=S_1$.} of \cite{Hassin2017} for scenarios 1 and 2 \blue{applicable to the case $\rho < 1$}.

\begin{table}[h]
\caption{Cases when $P_I=1$ is the equilibrium.} \label{tab:HR_all_custs}
\centering
\begin{tabular}{|c|c|c|}
\hline
 Scenario & Range of $R$ & Condition on $C_I$\\
 \hline & & \\ [-2.5ex]
 1 \& 2 & $R > c_w\left(\dfrac{2}{\mu}-\dfrac{1}{\mu+\lambda}\right) $ & $C_I\leq C_{I,1}^{1,2}(R)\equiv \dfrac{(1-\rho)\rho^{n_e}}{1-\rho^{n_e+1}}\left[\dfrac{c_wn_e}{\mu}+\dfrac{c_w}{\mu}-R\right]$ \\[1.7ex]
 3 & $R \in \left(\dfrac{c_w}{\mu},c_w\left(\dfrac{2}{\mu}-\dfrac{1}{\mu+\lambda}\right)\right]$ & $C_I\leq C_{I,1}^{3}(R) \equiv \dfrac{R\mu-c_w}{\mu+\lambda}$ \\[1.7ex]
 \hline
\end{tabular}
\end{table}

\blue{The functions $C_{I,1}^{3}(R)$ and $C_{I,1}^{1,2}(R)$ bound the region where all customers optimally inspect. In the remainder of this subsection, we will study properties of these functions to explain, in particular, the relationship between the discrete nature of the optimal threshold $n_e$ with the value of queue length information.

\paragraph{Properties of $C_{I,1}^{3}(R)$:} For $R>0$, $C_{I,1}^{3}(R)$ increases linearly in $R$ from $C_{I,1}^{3}(c_w/\mu)=0$.

\vspace{15pt}

The above properties of $C_{I,1}^{3}(R)$ are immediate from its formula in Table \ref{tab:HR_all_custs}. They make logical sense: when $R=c_w/\mu$ we have $n_e=0$, so inspecting the queue is useless. As $R$ increases, it becomes more and more worthwhile to join an empty queue, so the value in inspecting the queue over balking (the best alternative) increases. 

The following properties of $C_{I,1}^{1,2}(R)$ above are proved formally in Appendix \ref{append:proofs_all_custs}.

\paragraph{Properties of $C_{I,1}^{1,2}(R)$:}
For $R>c_w/\mu$, the function $C_{I,1}^{1,2}(R)$ is piecewise linear, with breakpoints at $xc_w/\mu$ for $x \in \{2,3,\ldots\}$, and each linear segment strictly decreasing. Further, for any $R>c_w/\mu$, we have $C_{I,1}^{1,2}(R)\geq 0$, with $C_{I,1}^{1,2}(R)=0$ if and only if $R=xc_w/\mu$ for some $x\in \{2,3\ldots \}$. The sequence
\begin{equation} \label{eqn:toothheight}
\left\{\lim_{\delta\downarrow 0} C_{I,1}^{1,2}\left(\frac{xc_w}{\mu}+\delta\right): x=2,3\ldots\right\}    
\end{equation}
is strictly decreasing with limit 0.

\vspace{15pt}

Since it also depends on the optimal threshold $n_e$, the function $C_{I,1}^{1,2}(R)$ has the same breakpoints as $C_{I,0}(R)$, which bounds the region where $P_I=0$.} However, unlike $C_{I,0}(R)$, we see in Figure \ref{fig:all} that $C_{I,1}^{1,2}(R)$ is not continuous in $R$, showing a sawtooth pattern, with the height of each tooth in \eqref{eqn:toothheight} decreasing as $n_e$ increases. \blue{To explain this behaviour,} note that $C_{I,1}^{1,2}(R)$ is the maximum value of $C_I$ at which all arrivals inspecting is an equilibrium. Suppose $R\in (c_wx/\mu,c_w(x+1)/\mu]$, so $n_e=x$ and inspecting arrivals will only join a queue of length $x-1$ or less. Therefore, if all arrivals inspect, the queue length must belong to the set $\{0,1,\ldots , x\}$. When $R=c_w(x+1)/\mu$, the reward is equal to the expected waiting cost when joining a length-$x$ queue. Therefore, a new arrival is indifferent between joining and not joining when the queue is length $x$, so the queue is never a length where joining is suboptimal. Hence, a new arrival has no need for queue-length information, so for any $C_I>0$, will not pay to inspect if everyone else has. In other words, $P_I=1$ ceases to be the equilibrium as soon as $C_I$ increases from 0, explaining why $C_{I,1}^{1,2}(R)$ takes value 0 at its breakpoints.

As $R$ drops from $c_w(x+1)/\mu$ to $c_wx/\mu$, the optimal threshold stays at $x$, but joining a queue of length $x$ becomes more and more unfavourable, so joining the queue blindly instead of inspecting carries more and more risk. This increases the value of queue-length information if everyone else has inspected. Hence, a new arrival is willing to pay a higher and higher price to inspect, meaning $P_I=1$ remains the equilibrium for larger and larger $C_I$.

As $R$ increases across breakpoints (increasing the value of $x$), for a fixed traffic intensity, the steady state probability of the queue being length $x$ when all customers inspect decreases, reducing the amplitude of the above effect. Therefore, the height of the teeth decreases with $R$, decaying to 0 as $n_e=\infty$ (where queue-length information is useless) is approached.

\blue{Figure \ref{fig:all} shows that $C_{I,1}^{3}(R)$ and $C_{I,1}^{1,2}(R)$ coincide as we pass from scenario 3 to scenario 2 (at $R=c_w(2/\mu -1/(\mu+\lambda))$), which is easily verified by their formulae. We summarise in the following.}

\begin{proposition} \label{thm:all_toll}
Write 
\begin{equation*}
C_{I,1}(R) \equiv \begin{cases}
C_{I,1}^{3}(R) \quad &\text{for }\frac{c_w}{\mu} \leq R \leq c_w\left(\frac{2}{\mu}-\frac{1}{\mu+\lambda}\right),\\
C_{I,1}^{1,2}(R) \quad &\text{for }R \geq c_w\left(\frac{2}{\mu}-\frac{1}{\mu+\lambda}\right).
\end{cases}
\end{equation*}
Then $C_{I,1}(R)$ is piecewise \blue{continuous} with breakpoints at $c_w\left(\frac{2}{\mu}-\frac{1}{\mu+\lambda}\right)$ and $xc_w/\mu$ for $x\in \{2,3\ldots\}$. The only discontinuities of $C_{I,1}(R)$ are at $xc_w/\mu$ for $x\in \{2,3\ldots\}$. Whenever $C_I \leq C_{I,1}(R)$, all customers inspect the queue at the unique equilibrium.
\end{proposition} 

\subsection{\blue{Equilibria Where} Some Customers Inspect} \label{sec:somecusts}
For those equilibria where only some customers inspect the queue (where $P_I \in (0,1)$) we split our analysis by scenario. \blue{In each, we make an advancement on the analysis of \cite{Hassin2017} for the case where $\rho<1$.}

\paragraph{Scenario 1} It follows from Lemma \ref{lem:Scen1_prop} that any equilibrium in scenario 1 must have $P_B=0$. Combined with the results in Tables \ref{tab:HR_no_custs} and \ref{tab:HR_all_custs}, when $C_{I,1}^{1,2}(R) < C_I < C_{I,0}^1(R)$, the equilibrium is $(1-P_J^*,P_J^*)$ for $P_J^* \in (0,1)$.

By definition, the equilibrium probability $P_J^*$ satisfies $U_{J-I}(1-P_J^*,P_J^*)=0$. If we note that $\rho_L=\rho$ when $P_B=0$ and recall that $\rho_U=P_J \rho$, then we may use \eqref{eqn:V_J-I_eqn_closed} to reduce this equation to the following quadratic in $y \equiv 1-P_J^*\rho$:
\begin{equation} \label{eqn:quad_simple}
y^2\left(C_I\frac{\rho^{-n_e}-1}{1-\rho}\right)+y\left(C_I+R-\frac{c_wn_e}{\mu}\right)-\frac{c_w}{\mu}=0.
\end{equation}
Note that both the quadratic and linear coefficients in \eqref{eqn:quad_simple} are strictly positive, the former since \blue{we assume} $\rho<1$ and the latter because $R\mu>c_wn_e$ by the definition of $n_e$. Clearly the constant term is negative. Therefore, the product of the roots is negative, meaning there is one positive and one negative root. Selecting the positive root from the quadratic formula obtains a closed form for $P_J^*$.

Note that, if $C_I=C_{I,1}^{1,2}(R)$ is substituted into \eqref{eqn:quad_simple}, then setting $y=1$ yields 0; therefore, $y=1$ is the positive root of \eqref{eqn:quad_simple} so $P_J^*=0$. Likewise, if we set $C_I=C_{I,0}^{1}(R)$ and $y=1-\rho$, then \eqref{eqn:quad_simple} becomes 0 from which it follows that $P_J^*=1$. 
In Section \ref{sec:SW_lblue} when social welfare is analysed, it will be shown that $P_J^*$ is increasing on the region $[C_{I,1}^{1,2}(R),C_{I,0}^{1}(R)]$ as it becomes more costly to inspect the queue. 

\paragraph{Scenario 2} Moving from scenario 1 to scenario 2, the reward $R$ drops such that balking may now form part of the equilibrium. \emph{Proposition 1} of \cite{Hassin2017} includes a value\footnote{Denoted $K_2$ on page 811 of \cite{Hassin2017}. We multiply by $c_w/\mu$ to obtain $C_{B,0}^2(R)$.} , which we denote $C_{B,0}^2(R)$, that divides the interval on which $P_I \in (0,1)$ into two: for $C_I \in (C_{I,1}^{1,2}(R), C_{B,0}^2(R)]$ we have $P_B=0$, whilst for $C_I \in (C_{B,0}^{2}(R), C_{I,0}^{2,3}(R))$ all three equilibrium probabilities are non-zero. However, the form of $C_{B,0}^2(R)$ is extremely complicated, involving several layers of nested fractions involving polynomials of degree $n_e$. 
In the following, \blue{for the case where $\rho < 1$,} we offer a much simpler function which bounds the region where $P_B=0$.

\begin{corollary} \label{corol:bound_on_p*,1}
Write
\begin{equation} \label{eqn:bound_C_J,1^2(R)}
\widehat{C}_{B,0}^2(R)\equiv R(1-\rho^{n_e})+\frac{c_w((n_e+1)\rho^{n_e}-n_e\rho^{n_e+1}-1)}{\mu-\lambda}.
\end{equation}
In scenario 2, whenever we have
$$C_I \in (C_{I,1}^{1,2}(R),\widehat{C}_{B,0}^2(R)),$$ 
the equilibrium has $P_B=0$.
\end{corollary}

\begin{proof}
Note that if $U_I(P_I,P_J)>0$ for all $P_I$ and $P_J$, then balking is never a best response, 
so the equilibrium must have the form $P_B=0$. By Lemma \ref{lem:mono_all_U}, we have $U_I(P_I,P_J)\geq U_I(1-P_J,P_J) \geq U_I(0,1)$, so 
\blue{we must have $U_I(P_I,P_J)>0$ for all $P_I$ and $P_J$ if $U_I(0,1)>0$}. Using \eqref{eqn:VI_0,p_J} with $P_J=1$, \blue{the inequality $U_I(0,1)>0$} is equivalent to $C_I < \widehat{C}_{B,0}^2(R)$. \blue{We conclude that} the equilibrium must have $P_B=0$ under this condition. 
\end{proof}

Figure \ref{fig:all} shows both $C_{B,0}^2(R)$ and $\widehat{C}_{B,0}^2(R)$, the latter as a dotted line. It can be seen that the bound $\widehat{C}_{B,0}^2(R)$ determines the majority of scenario 2 cases with $P_B=0$ in example of Figure \ref{fig:all}, with the same seen in many other numerical examples we checked.

We see from Figure \ref{fig:all} that both $C_{B,0}^2(R)$ and $\widehat{C}_{B,0}^2(R)$ are increasing in $R$, a property which makes logical sense, since the larger the reward, the more inclined an arrival is to join with certainty if they do not inspect the queue. Indeed, it is clear from \eqref{eqn:bound_C_J,1^2(R)} that $\widehat{C}_{B,0}^2(R)$ increases linearly in $R$ for fixed $n_e$.

As for the values of the equilibrium probabilities in scenario 2, when $P_B=0$, like scenario 1, the equilibrium is $(1-P_J^*,P_J^*)$ where $y=1-P_J^* \rho$ is the positive root of \eqref{eqn:quad_simple}. When all three equilibrium probabilities are non-zero, the equilibrium satisfies $U_I(P_I,P_J)=U_J(P_I,P_J)=0$, for which no closed form for the equilibrium probabilities $P_I$ and $P_J$ could be found.

\paragraph{Scenario 3}
Moving into scenario 3, by Proposition \ref{lem:Scen3_prop} balking now becomes the best response to $P_I=1$ instead of blindly joining. Consequently, the region with $P_I \in (0,1)$ is now divided into $P_J=0$ and $P_I,P_J,P_B>0$.
\emph{Proposition 1} of \cite{Hassin2017} includes a polynomial of degree $n_e+1$ to solve to obtain an inequality\footnote{Denoted $\kappa \leq K_3$ on page 811 of \cite{Hassin2017}.} defining this division.
\blue{For the $\rho < 1$ case}, we may use Proposition \ref{lem:Scen3_prop} to set $n_e=1$ and forgo the need for the polynomial, greatly simplifying the dividing inequality. In addition, our method provides a simple closed formula for the equilibrium probability $P_I^*$ when $P_J=0$.

\begin{proposition} \label{prop:pI,0:class}
In Scenario 3, when
\begin{equation} \label{eqn:defn_CJO}
C_{I,1}^{3}(R)\leq C_I \leq C_{J,0}^{3}(R)\equiv \dfrac{(R\mu-c_w)(2c_w-R\mu)}{\mu},
\end{equation}
the unique equilibrium is $(P_I^*,0)$ where
\begin{equation} \label{eqn:defn_PI*}
P_I^* \equiv \dfrac{R\mu-c_w-C_I\mu}{C_I \lambda}.
\end{equation}
\end{proposition}

\begin{proof}
Suppose we are in Scenario 3. Then, by Proposition \ref{lem:Scen3_prop}, we must have $n_e=1$, so \eqref{eqn:VJ_form_gen} simplifies to 
\begin{equation} \label{eqn:VJ_l=1}
U_J(P_I,P_J)=R-\frac{c_w}{\mu}\left[1+\frac{\rho(P_I+P_J)}{(1-\rho P_J)(1+\rho P_I)}\right],
\end{equation}
and \eqref{eqn:VI_form_gen} simplifies to 
\begin{equation} \label{eqn:VI_l=1}
U_I(P_I,P_J)=\left(R-\frac{c_w}{\mu}\right)\left(\frac{1-\rho P_J}{1+\rho P_I}\right) -C_I.
\end{equation}

Now let $P_I \in (0,1)$. By definition we have $(P_I,0)$ an equilibrium if and only if $U_I(P_I,0)=0\geq U_J(P_I,0)$. By rearrangement of \eqref{eqn:VI_l=1} with $P_J=0$, the unique $P_I$ satisfying $U_I(P_I,0)=0$ is $P_I^*$ in \eqref{eqn:defn_PI*}. To be an equilibrium, we need both $P_I^* \in (0,1)$ and $U_J(P_I^*,0)\leq 0$. We next find the region where both of these conditions hold.

First, recall the definition of $C_{I,1}^3(R)$ in Table \ref{tab:HR_all_custs}. Using the formula for $P_I^*$ in \eqref{eqn:defn_PI*}, it is easy to see that 
\begin{equation} \label{eqn:cond_p_I_0_a_prob}
P_I^* \in (0,1) \quad \text{if and only if} \quad C_{I,1}^3(R) < C_I < \frac{R\mu-c_w}{\mu}.
\end{equation}
Second, by \eqref{eqn:VJ_l=1}, we have
\begin{align*}
U_J(P_I^*,0)&=R-\frac{c_w}{\mu}\left[1+\frac{\rho P_I^*}{1+\rho P_I^*}\right] \\
&=R-\frac{c_w}{\mu}\left[2-\frac{1}{1+\rho P_I^*}\right] \\
&=R-\frac{2c_w}{\mu}+\frac{c_wC_I}{R\mu-c_w},
\end{align*}
which is non positive if any only if 
\begin{equation} \label{eqn:t_cond_p_I_0}
C_I\leq\frac{(R\mu-c_w)(2c_w-R\mu)}{\mu},
\end{equation}
where the right-hand side is $C_{J,0}^3(R)$ defined in \eqref{eqn:defn_CJO}. We require both \eqref{eqn:cond_p_I_0_a_prob} and \eqref{eqn:t_cond_p_I_0} to hold for $(P_I^*,0)$ to be the equilibrium. Since $n_e=1$, we have $2c_w-R\mu\in (0,1)$, so the upper bound in \eqref{eqn:t_cond_p_I_0} is binding. This completes the proof.
\end{proof}

It is clear from the formula in \eqref{eqn:defn_CJO} that $C_{J,0}^3(c_w/\mu)=0$; it is also easy to verify using \eqref{eqn:defn_CJO} that $C_{J,0}^3(R)$ coincides with both $C_{I,1}^3(R)$ and $C_{I,1}^{1,2}(R)$ at $R=c_w(2/\mu-1/(\mu+\lambda))$ when scenario 3 becomes scenario 2. Both of these properties can be seen in Figure \ref{fig:all}. Note also from Figure \ref{fig:all} that the proportion of the region $(C_{I,1}^3(R),C_{I,0}^3(R))$ with $P_J=0$ decreases from 1 to 0 as $R$ increases across scenario 3, a logical observation since the larger the reward, the more incentive to a blind customer to join over balk.

In a similar manner, the equilibrium probability $P_I^*$ in \eqref{eqn:defn_PI*} is clearly increasing in $R$ due to balking waning in attractiveness compared to joining. To find the equilibrium probabilities in the region where all three equilibrium probabilities are non zero, like scenario 2, one must solve $U_I(P_I,P_J)=U_J(P_I,P_J)=0$ with no known closed solution. However, in the example in Figure \ref{fig:all}, such an equilibrium involving all three pure strategies is quite rare; we observed the same behaviour in many other examples we tested. Hence, across Section \ref{sec:find_eq}, we have found closed expressions for the equilibrium probabilities in the majority of cases.

\section{Social Welfare as $R$ and $C_I$ Vary} \label{sec:SW}
Whilst in Section \ref{sec:find_eq} we examined the structure of the equilibrium, in this section, we will study utility at the equilibrium, which we call the \emph{social welfare.} In \cite{Hassin2017}, it is shown that social welfare is decreasing in $\kappa \equiv C_Ic_w/\mu$.  \blue{For the case $\rho<1$,} our analysis will compare directly the effects of marginal increases to $R$ and marginal decreases to $C_I$ on social welfare, motivated by a service provider who has the choice of the two to invest in.
We will show that increasing $R$ is only the most effective investment when everyone already joins the queue with certainty at the equilibrium. Otherwise, decreasing $C_I$ has a greater effect on social welfare. Further, whilst decreasing $C_I$ always increases social welfare, we show that, in the case where \emph{some} arrivals inspect the queue at equilibrium, increasing $R$ can sometimes, paradoxically, decrease social welfare. We draw parallels to Braess's paradox \cite{Braess1968}, where adding more roads to a network can slow down overall traffic flow.

Throughout, we will use the observation that if the equilibrium satisfies $P_X>0$ for any $X\in\{I,J,B\}$, then the social welfare is equal to the utility $U_X$ evaluated at the equilibrium. 
Therefore, if $P_B>0$ at the equilibrium, the social welfare is 0 since balkers always have utility 0; in Figure \ref{fig:all},
the yellow, brown and pink regions are those with $P_B>0$. If $P_B=0$, then the social welfare must be strictly positive, since there must be no incentive for an arrival to balk at the equilibrium. The corresponding regions in Figure \ref{fig:all} 
are green where $P_J=1$, blue where $P_I=1$, and cyan where $P_I,P_J<1$. We shall examine the social welfare in each of these regions individually in Sections \ref{sec:SW_green}, \ref{sec:SW_dblue} and \ref{sec:SW_lblue}, before summarising and providing a contour plot in Section \ref{sec:summarySW}.

\subsection{Green Region: Equilibrium $P_J=1$} \label{sec:SW_green}
The green region of Figure \ref{fig:all} is easy to analyse. The equilibrium is $P_J=1$, so we may use \eqref{eqn:VI_0,p_J} to deduce that the social welfare is  
\begin{equation} \label{eqn:SW_green_region}
U_J(0,1)=R-\frac{c_w}{\mu-\lambda}.
\end{equation}
We see that a unit increase in $R$ generates a unit increase in social welfare, a logical result since all arrivals blindly join at equilibrium, and any blind joiner benefits from an increase in $R$. 
Decreasing $C_I$ has no effect on social welfare since $C_I$ is too high for any arrival to inspect at equilibrium. Therefore, when $P_J=1$ at equilibrium, the service provider sees a better immediate rate of improvement by increasing $R$ than decreasing $C_I$. We will next show that the reverse is true in the blue and cyan regions, where the equilibrium satisfies $P_J<1$.

\subsection{Blue Region: Equilibrium $P_I=1$} \label{sec:SW_dblue}
In the blue region of Figure \ref{fig:all}, the equilibrium is $P_I=1$, so the social welfare is $U_I(1,0)$ expressed in \eqref{eqn:VI_1,*}. By its definition in \eqref{eqn:ne_def}, the optimal threshold $n_e$ is unaffected by decreasing $C_I$. It is therefore clear from \eqref{eqn:VI_1,*} that a unit decrease in $C_I$ leads to a unit increase in social welfare. 

As for increasing $R$, we must take into account that $n_e$ increases by 1 when $R$ passes above $(n_e+1)c_w/\mu$. 
However, since $C_I>0$, it follows from the properties of $C^{1,2}_{I,1}(R)$ discussed in Section \ref{sec:all_custs} (and is clearly visible in Figure \ref{fig:all}) that, in order for $n_e$ to change, we \emph{must} leave the  blue region. Therefore, any increase of $R$ remaining in the blue region does not increase $n_e$. Since $\rho<1$, inspection of \eqref{eqn:VI_1,*} shows that a unit increase in $R$ leads to an increase of $(1-\rho^{n_e})/(1-\rho^{n_e+1})<1$ in social welfare, which is precisely equal to $P[Q<n_e]$ when $P_I=1$.

We conclude that, in the blue region, decreasing $C_I$ sees a greater increase in social welfare than increasing $R$. This observation is logical. All arrivals inspect at equilibrium, and any inspector benefits from a reduction in $C_I$. However, only if an inspector joins the queue with probability $P[Q<n_e]$ will they benefit from an increase in $R$. As $n_e \rightarrow \infty$, we have $P[Q<n_e] \rightarrow 1$, and, since all inspectors now join the queue, increasing $R$ becomes just as beneficial as decreasing $C_I$.

\subsection{Cyan Region: Equilibrium $(1-P_J^*, P_J^*)$} \label{sec:SW_lblue}
In Section \ref{sec:somecusts}, we deduced that the equilibrium in the cyan region was $(1-P_J^*, P_J^*)$, where $y^*\equiv 1-P_J^*\rho$ is the positive root of the quadratic equation in \eqref{eqn:quad_simple} which we restate below: 
\begin{equation} \label{eqn:quad_simple_S5}
y^2\left(C_I\frac{\rho^{-n_e}-1}{1-\rho}\right)+y\left(C_I+R-\frac{c_wn_e}{\mu}\right)-\frac{c_w}{\mu}=0.
\end{equation}
Note that the coefficients in \eqref{eqn:quad_simple_S5} (and hence $P_J^*$ itself) depend on $C_I$ and $R$, which makes the analysis of this case more complicated than the green and blue regions on which the equilibrium probabilities were constant.

Recall from Section \ref{sec:somecusts} that the quadratic and linear coefficients in \eqref{eqn:quad_simple_S5} are strictly positive and the constant term is negative, leading to a unique positive root. To help our analysis, we first state a lemma about quadratics of this type, the proof of which can be found in Appendix \ref{append:proof:lem:gen_quad}.

\begin{lemma} \label{lem:gen_quad}
Consider a general quadratic equation with coefficients (in descending order of powers) given by $a>0$, $b>0$ and $c<0$. If either $a$ or $b$ is increased, the value of the positive root must decrease.
\end{lemma}

Applied to our quadratic in \eqref{eqn:quad_simple_S5}, Lemma \ref{lem:gen_quad} shows that the positive root $y^*=1-P^*_J\rho$ is decreased both by an increase in $C_I$ and by an increase in $R$ small enough to not change the threshold $n_e$. Therefore, decreasing $C_I$ leads to an decrease in $P_J^*$ and increasing $R$ leads to an increase in $P_J^*$. We now apply this result to the social welfare $U_I(1-P_J^*,P_J^*)=U_J(1-P_J^*,P_J^*)$.

\paragraph{Decreasing $C_I$}
By \eqref{eqn:VI_form_gen}, we see that the social welfare $U_I(1-P_J^*,P_J^*)$ is of the form $f(P_J^*)-C_I$ where $f$ does not involve $C_I$. By Lemma \ref{lem:mono_all_U}, we have $U_I(1-P_J,P_J)$ strictly decreasing in $P_J$ when $C_I$ is held constant. Since decreasing $C_I$ sees $P_J^*$ decrease, it follows that a unit decrease in $C_I$ sees a strictly greater than unit increase in social welfare.

This result is logical for the following reason. In the cyan region, there are queue inspectors at the equilibrium who benefit directly (via a unit increase in social welfare) from a unit reduction in $C_I$. Also, $C_I$ decreasing leads to $P_J^*$ decreasing, which reduces the traffic intensity, $\rho_U$, above the threshold (the traffic intensity below the threshold, $\rho_L$, is fixed at $\rho$). Therefore, inspectors who go on to join the queue \emph{additionally} benefit from lower congestion, leading to an overall greater than unit increase in social welfare. 

\paragraph{Increasing $R$ without Increasing $n_e$}
Similarly to the case of decreasing $C_I$, we note from \eqref{eqn:VJ_form_gen} that the social welfare $U_J(1-P_J^*,P_J^*)$ takes the form $R-g(P_J^*)$ where $g$ does not involve $R$. By Lemma \ref{lem:mono_all_U}, we have $U_J(1-P_J,P_J)$ strictly decreasing in $P_J$ when $R$ is held constant. Since increasing $R$ sees $P_J^*$ increase, it follows that a unit increase in $R$ sees a change in social welfare of no more than +1. The following result, proved in Appendix \ref{append:proof:prop:inc_rew_no_thres_cross}, shows that the change in social welfare belongs to the interval $(0,1)$.

\begin{theorem} \label{prop:inc_rew_no_thres_cross}
Any increase in the reward $R$ that does not increase the threshold $n_e$ leads to an increase in social welfare.
\end{theorem}

A less-than-unit increase in social welfare from increasing $R$ in the cyan region is again logical. There are blind joiners at the equilibrium who benefit directly (via a unit increase in social welfare) from a unit increase in $R$. However, increasing $R$ leads to $P_J^*$ increasing, which increases $\rho_U$ whilst keeping $\rho_L=\rho$. Therefore, blind joiners also face increased congestion, which detriments social welfare. Theorem \ref{prop:inc_rew_no_thres_cross} shows that the beneficial increase in reward received is more significant than the extra congestion.

\paragraph{Increasing $R$ to increase $n_e$}
We now consider increases to $R$ in the cyan region which do increase the optimal threshold $n_e$. We first examine the effect on $P_J^*$ at the moment $n_e$ increases.

\begin{lemma} \label{prop:equil_prob_cross_threshold}
When the equilibrium satisfies $P_B=0$ and $P_J,P_I<1$, for any $x \in \{2,3,\ldots\}$, there exists $\epsilon_x>0$ such that, for any $\epsilon \in (0,\epsilon_x)$, increasing $R$ from $xc_w/\mu$ to $xc_w/\mu + \epsilon$ decreases the equilibrium probability $P_J^*$.
\end{lemma}

Lemma \ref{prop:equil_prob_cross_threshold}, whose proof can be found in Appendix \ref{append:proof:prop:equil_prob_cross_threshold}, shows that when $R$ passes $xc_w/\mu$ and $n_e$ increases from $x$ to $x+1$, some blind joiners are replaced by inspectors at the equilibrium, which reduces queue congestion. However, the increase in $n_e$ means that any inspector now joins a queue of length $x$ instead of balking, which increases queue congestion.
In the next section, we will show in Figure \ref{fig:contour} that the latter can be more significant. In other words, the social welfare can decrease when $R$ is increased and a threshold is crossed because of increased queue congestion. Hence Theorem \ref{prop:inc_rew_no_thres_cross} cannot be extended to the case where $n_e$ increases. \blue{We summarise in the following.

\begin{corollary} \label{corol:incRdecSW}
It is possible for an increase in the reward $R$ to decrease social welfare if the threshold $n_e$ increases as a result. 
\end{corollary}}

In fact, our numerical experiments failed to find a single case where social welfare does not fall as $R$ increases over a threshold, suggesting a \blue{stronger} analytical result \blue{than Corollary \ref{corol:incRdecSW}} may be possible. However, despite our best efforts, comparing the complex expression for $U_I(1-P_J^*,P_J^*)$ (found in equation \eqref{eqn:V_I_nobalk_at_equil}) either side of the threshold is difficult due to the increase in $n_e$ when the threshold is passed. We therefore leave such an investigation as further work.

\subsection{Summary and Social Welfare Sensitivity} \label{sec:summarySW}
In this subsection, we will summarise our results on the effects of marginal changes to $R$ and $C_I$ on social welfare, before discussing their practical implications. First, we reproduce Figure \ref{fig:all} as a contour plot of social welfare in Figure \ref{fig:contour}, with scales adjusted to be the same for each axis. 

\begin{figure}
\centering
\includegraphics[width=1\textwidth]{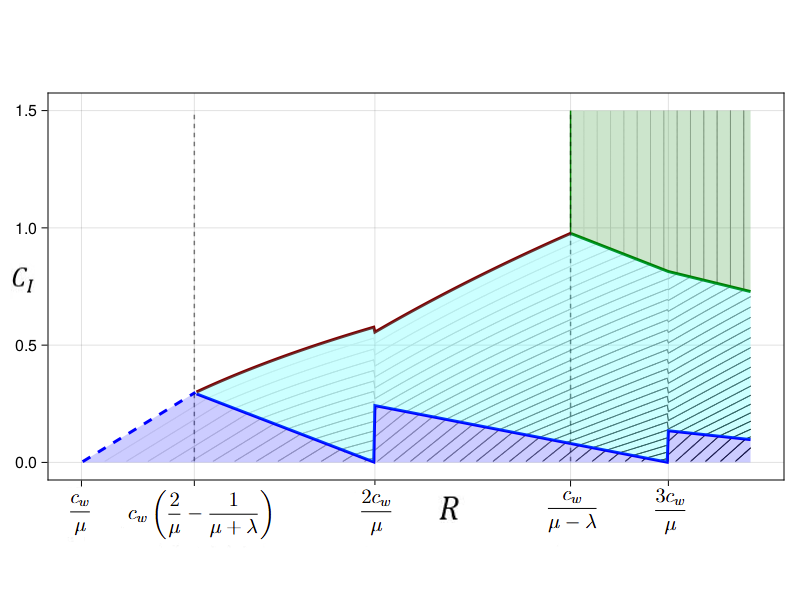} 
\caption{(Best viewed in colour) For the example in Figure \ref{fig:all} with $\lambda=0.5, \mu=0.8$ and $c_w=1$, a contour plot showing changes to social welfare as $R$ and $C_I$ vary. As in Figure \ref{fig:all}, the increments of $c_w/\mu$ on the $R$-axis demonstrate the regions where $n_e$ is constant. Only those coloured regions where the social welfare is non zero (namely green, blue and cyan) are retained from Figure \ref{fig:all}. The darker the contour lines, the greater the social welfare.} \label{fig:contour}
\end{figure}

In the green region, the contour lines are vertical, which, as deduced in Section \ref{sec:SW_green}, shows that only increasing $R$ increases social welfare. In the blue and cyan regions, when a threshold line $xc_w/\mu$ is not crossed, the contour lines are slanted forwards, indicating that both increasing $R$ and decreasing $C_I$ increase social welfare. Since the axes have the same scale, the contour lines having a \blue{slope} of less than 45 degrees shows that a marginal decrease in $C_I$ has the greater effect, as proven in Sections \ref{sec:SW_dblue} and \ref{sec:SW_lblue}.

In the blue region, as also shown in Section \ref{sec:SW_dblue}, we see that as the threshold $n_e$ increases, the \blue{slope of the contour lines} approaches 45 degrees, where changes in $R$ and $C_I$ affect social welfare at the same rate. We see the same effect in the cyan region, \blue{a phenomenon explained in Section \ref{sec:SW_lblue}: decreasing $C_I$ in the cyan region decreases the ratio $P^*_J/P^*_I$ of blind joiners to inspectors (reducing congestion) whilst increasing $R$ increases the same ratio (increasing congestion). As a result, decreasing $C_I$ is more beneficial than increasing $R$.}

However, as $n_e$ increases, a larger proportion of inspectors will join the queue, so the behaviour of inspectors `converges' to that of blind joiners. Therefore, any change to $P^*_J/P^*_I$ has a smaller and smaller effect on congestion as $n_e$ grows.

We also see in Figure \ref{fig:contour} that, for the same $n_e$, the contour \blue{slopes} are closer to 45 degrees in the blue region than the cyan region. This effect was prevalent in all numerical examples we investigated, and shows that, for fixed $n_e$, the extra benefit in decreasing $C_I$ over increasing $R$ is greater in the cyan region than in the blue. This indicates that the changes in congestion which favour decreasing $C_I$ in the cyan region are more significant than increases to $R$ not benefiting inspectors who balk in the blue region. 

Finally, we note that, whenever a contour line crosses a line $R=xc_w/\mu$, there is a sudden drop in the value of social welfare. As discussed in Section \ref{sec:SW_lblue}, this effect occurs because the threshold $n_e$ increases so a higher proportion of inspectors now join the queue, and was, in fact, present in all numerical examples we investigated.

We conclude with some practical implications of the above findings to a service provider wishing to improve social welfare through investment. First, the current values of $C_I$ and $R$ should be located on Figure \ref{fig:contour}.
If in the green region, for quick results, investment should go towards increasing $R$, since decreasing $C_I$ has no immediate benefit. However, if increasing $R$ and decreasing $C_I$ come at the same cost, a longer-term investment may be better suited on $C_I$, as if the cyan region is reached, the additional benefit of reduced congestion is attained by decreasing $C_I$. 

In the cyan region, Figure \ref{fig:contour} shows that increasing $R$ can decrease social welfare when the optimal threshold increases, suggesting that the service provider should be cautious whenever increasing $R$ here. Yet, in real life, it is unlikely that all queue inspectors will follow the exact optimal policy and all begin to join a queue of length $x$ at the precise moment that $R$ passes $xc_w/\mu$. In practice, the increases in congestion due to more inspectors joining the queue will come about much more gradually. Nevertheless, this effect, combined with an increase in blind joiners at equilibrium, does render increasing $R$ less beneficial than decreasing $C_I$ in the cyan region, assuming that both investments can be made at the same cost.

In the blue region, decreasing $C_I$ is also more beneficial than increasing $R$, albeit less so than in the cyan region. Yet, $C_I$ is already quite low here, and it may be practically difficult to decrease it further. For example, no matter how well designed and visible the queue length is on a service provider's website, an arrival will still have to wait and/or pay to load the webpage, so some inspection cost is unavoidable. Therefore, increasing $R$ for a smaller benefit could be the only option here.

Outside the green, blue and cyan regions of Figure \ref{fig:contour}, there is no immediate benefit to increasing $R$ or decreasing $C_I$; the social welfare will remain at 0. Therefore, the service provider should investigate which action will fastest (or cheapest) take them into the green, blue and cyan regions where social benefit improvements can be made.

\section{Conclusion and Future Research} \label{sec:conclusion}
The work \cite{Hassin2017} proposes a novel variation on an M/M/1 queue where customers can pay to learn the queue length, a model of increasing relevance due to more and more queue operators providing this information on websites or apps. Our work builds on the solutions of \cite{Hassin2017} by refining the classification of equilibrium types, and, in many cases, finding simpler closed-form expressions for the equilibrium probabilities. We also reframe the results in terms of the inspection cost and service reward, the two quantities most easily influenced by a change in business strategy. Our finding that altering inspection cost often influences social welfare more than altering service reward not only shows the importance of information in queueing, but is also of great practical use to the queue operator. If the aim is to maximise social welfare, it motivates making queue-length information easily accessible, possibly by displaying it on a homepage, or building a website quick to load on portable devices.

We now discuss future work. As in most studies on strategic queueing, with very few exceptions (e.g., \cite{Kerner2011}), we have assumed the basic M/M/1-type model. Considering non-Markovian queueing models and in particular queueing models with a general service time distribution is very desirable, but at the same time looks to be extremely challenging analytically. \blue{In addition, we consider a time-stationary arrival rate $\lambda$ always less than the service rate $\mu$. In practice, the arrival rate often varies throughout different periods of the day, rising---possibly above the service rate---at certain peak times. Therefore, an extension to a time-dependent arrival rate would certainly benefit the literature. However, no matter the choice of model or arrival rate, in real life} customers are often ignorant of the true system parameters. Thus, it is very important to develop learning schemes \cite{Massaro2019,Qian2019} or to analyse various behavioural approaches \cite{Young1993}. 

As discussed at the end of Section \ref{sec:SW_lblue}, we leave open the question of whether the social welfare \emph{always} drops in value at the moment when increasing the reward causes an increase in the optimal threshold, \blue{which would strengthen Corollary \ref{corol:incRdecSW}}. Whilst we tried hard to find either a proof or counter-example, we were not able to succeed in these directions.

Another formulation related to our model involves admission control with incomplete queue-length information---the works \cite{Kuri1995, Kuri1997} study this model in discrete time with a $k$ time step delay in queue-length observation. In contrast to our game-theoretic formulation, the admission-control problem is significantly harder due to the complexity of the state space. The work \cite{lin2003} adds multiple servers to admission control, but the state space is simplified by removing the queue, with an arrival forced to leave the system (for a cost) if all servers are busy.
The elaboration of the optimal control, reinforcement learning or multi-agent learning algorithms is possible and very important for practice.

\clearpage

\begin{appendices}
\section{Proofs for Section \ref{sec:no_custs}} \label{append:proofs_no_custs}
\paragraph{Proof of properties of $C_{I,0}^1(R)$}

\begin{proof}
It follows from the definition of $n_e$ in \eqref{eqn:ne_def} that $C_{I,0}^1(R)$ is piecewise \blue{differentiable} in $R$ with the stated breakpoints. Observation of the definition of $C_{I,0}^1(R)$ in Table \ref{tab:HR_no_custs} shows that each \blue{piecewise} component is linear and strictly decreasing. The continuity and overall strictly decreasing property of $C_{I,0}^1(R)$ will hence follow from showing that $C_{I,0}^1(R)$ is continuous at all of its breakpoints.

For any $x \in \{2,3,\ldots\}$, if $R=xc_w/\mu$ then $n_e=x-1$, so we have
\begin{equation} \label{eqn:t_0_breakpoint}
C_{I,0}^1(xc_w/\mu)=\rho^{x-1}c_w\left[\frac{1}{\mu-\lambda}-\frac{1}{\mu}\right].
\end{equation}
As $R$ tends to $xc_w/\mu$ from below, we still have $n_e=x-1$, so the left continuity of $C_{I,0}^1$ at $xc_w/\mu$ is immediate. To show right continuity, we consider $R=xc_w/\mu+\delta$ for $\delta \in (0,c_w/\mu)$. Since $n_e=x$ on this range, we have
\begin{equation} \label{eqn:t_0_breakpoint_limit}
\lim_{\delta \downarrow 0} C_{I,0}^1(xc_w/\mu+\delta)= \lim_{\delta \downarrow 0} \rho^{x}c_w\left[\frac{1}{\mu-\lambda}-\delta\right]  = \frac{\rho^{x}c_w}{\mu-\lambda}.
\end{equation}
Since
$$\frac{\rho}{\mu-\lambda}=\frac{1}{\mu-\lambda}-\frac{1}{\mu},$$
it follows that \eqref{eqn:t_0_breakpoint} and \eqref{eqn:t_0_breakpoint_limit} are equal and that $C_{I,0}^1$ is right continuous at $xc_w/\mu$.

To see that $C_{I,0}^1(R)$ is strictly positive for all $R$, note that we have
$$C_{I,0}^1(R)> \rho^{n_e} \left[ \frac{c_w(n_e+1)}{\mu}-R\right]\geq 0,$$
where the second inequality follows from the definition of $n_e$.

Finally, to prove the limit, note that $n_e\rightarrow \infty$ as $R\rightarrow \infty$. Further, by the definition of $n_e$, we have $c_wn_e/\mu-R<0$, so
$$C_{I,0}^1(R) < \frac{\rho^{n_e}c_w}{\mu-\lambda} \downarrow 0 \quad \text{as} \quad n_e \rightarrow \infty.$$ 
The proof is completed.
\end{proof}

\paragraph{Proof of properties of $C_{I,0}^{2,3}(R)$:}

\begin{proof}
It is clear from the definition of $n_e$ in \eqref{eqn:ne_def} that $C_{I,0}^{2,3}(R)$ is piecewise \blue{differentiable} in $R$ with the stated breakpoints. Since we assume $R\mu >c_w$, it is further clear from the definition of $C_{I,0}^{2,3}(R)$ in Table \ref{tab:HR_no_custs} that $C_{I,0}^{2,3}(R)$ is strictly positive and strictly increasing in $R$ for any fixed value of $n_e$. The continuity and overall strictly increasing property of $C_{I,0}^{2,3}(R)$ will hence follow from showing that $C_{I,0}^{2,3}(R)$ is continuous at all of its breakpoints.

For any $x \in \{2,3,\ldots\}$, if $R=xc_w/\mu$ then $n_e=x-1$, so we have
\begin{align} 
C_{I,0}^{2,3}(xc_w/\mu)&=\frac{(x-1)c_w}{\mu}\left(1-\frac{1}{x}\right)^{x-1} \nonumber \\
&=\frac{xc_w}{\mu}\left(1-\frac{1}{x}\right)^{x}. \label{eqn:CI023proof}
\end{align}
As $R$ tends to $xc_w/\mu$ from below, we still have $n_e=x-1$, so the left continuity of $C_{I,0}^{2,3}$ at $xc_w/\mu$ is immediate. To show right continuity, we consider $R=xc_w/\mu+\delta$ for $\delta \in (0,c_w/\mu)$. Since $n_e=x$ on this range, we have
\begin{equation*}
\lim_{\delta \downarrow 0} C_{I,0}^{2,3}(xc_w/\mu+\delta)= \lim_{\delta \downarrow 0} \frac{xc_w}{\mu}\left(1-\frac{c_w}{xc_w+\delta \mu}\right)^{x} = \frac{xc_w}{\mu}\left(1-\frac{1}{x}\right)^{x},
\end{equation*}
which is equal to \eqref{eqn:CI023proof}, so the right continuity of $C_{I,0}^{2,3}$ at $xc_w/\mu$ follows.
\end{proof}

\section{Proofs for Section \ref{sec:all_custs}} \label{append:proofs_all_custs}

\paragraph{Proof of properties of $C_{I,1}^{1,2}(R)$:}

\begin{proof}
As discussed previously, it is clear from the definition of $n_e$ in \eqref{eqn:ne_def} that $C_{I,1}^{1,2}(R)$ is piecewise \blue{differentiable} in $R$ with the stated breakpoints. Observation of the definition of $C_{I,1}^{1,2}(R)$ in Table \ref{tab:HR_all_custs} shows that each component is linear and strictly decreasing. The positivity of $C_{I,1}^{1,2}$ and the location of its zeros follows from the inequality 
$$\frac{c_w(n_e+1)}{\mu}-R \geq 0,$$
which again arises from the definition of $n_e$ in \eqref{eqn:ne_def}.

Finally, for $0<\delta \leq c_w/\mu$, we have $n_e=x$ when $R=xc_w/\mu +\delta$, so
$$\frac{(1-\rho)\rho^x}{(1-\rho^{x+1})}\left[\frac{c_w}{\mu}-\delta\right] \rightarrow \frac{c_w(1-\rho)\rho^x}{\mu(1-\rho^{x+1})} \quad \text{as} \quad \delta \downarrow 0.$$
Since $0<\rho<1$, the sequence 
$$\left\{\frac{c_w(1-\rho)\rho^x}{\mu(1-\rho^{x+1})}: x=2,3\ldots\right\}$$
is strictly decreasing with limit 0, completing the proof.
\end{proof}

\section{Proofs for Section \ref{sec:SW}}
\subsection{Proof of Lemma \ref{lem:gen_quad}} \label{append:proof:lem:gen_quad}
\begin{proof}
Since $ac<0$ is the product of the roots, the quadratic has a unique positive root, $\widehat{y}$. For any $\delta_a,\delta_b \geq 0$, we have
$$a\widehat{y}^2+b\widehat{y}+c=0, \quad (a+\delta_a)\widehat{y}^2+b\widehat{y}+c\geq 0 \quad \text{and} \quad a\widehat{y}^2+(b+\delta_b)\widehat{y}+c\geq 0.$$
Since $a, b > 0$, we have
\begin{equation} \label{eqn:proof_two_quads}
(a+\delta_a)y^2+by+c > 0 \quad \text{and} \quad ay^2+(b+\delta_b)y+c> 0
\end{equation}
for any $y \in (\widehat{y},\infty)$. Therefore, the two quadratics in \eqref{eqn:proof_two_quads} must have their unique positive root in $[0,\widehat{y}]$, completing the proof. 
\end{proof}

\subsection{Proof of Theorem \ref{prop:inc_rew_no_thres_cross}} \label{append:proof:prop:inc_rew_no_thres_cross}
\begin{proof}
When $P_B>0$ the social welfare is 0, and the cases $P_J=1$ and $P_I=1$ have already been proven in Sections \ref{sec:SW_green} and \ref{sec:SW_dblue}, respectively. Here, we prove the remaining case where $P_I,P_J<1$ and $P_B=0$, that in cyan in Figure \ref{fig:all}.

Write $P_J^*$ for the equilibrium probability of blindly joining the queue in this case. The social welfare at equilibrium is $U_I(1-P_J^*,P_J^*)=U_J(1-P_J^*,P_J^*)$. To complete the proof, we will show that $U_I(1-P_J^*,P_J^*)$ is increasing as a function of $R$.

We begin by deriving a closed form expression for $U_I(1-P_J^*,P_J^*)$. First, recall that $\rho_U\equiv P_J \rho$, and note that when $P_B=0$ we have $\rho_L=\rho$. Now consider \eqref{eqn:VI_form_gen} under these equalities. The bracketed terms can be written as $R(1-\rho^{n_e})-b$, where 
$$b \equiv \frac{c_w(1+n_e\rho^{n_e+1}-(n_e+1)\rho^{n_e})}{\mu(1-\rho)}>0$$
does not depend on $R$. Where $\rho_\Delta \equiv \rho_L-\rho_U$, we also have
$$\frac{(1-P_J \rho)}{(1-P_J \rho-\rho_\Delta \rho^{n_e})} = \left(1-\frac{\rho_\Delta \rho^{n_e}}{1-P_J \rho}\right)^{-1} = \left(1-\rho^{n_e}\left[1-\frac{1-\rho}{1-P_J \rho} \right]\right)^{-1}.$$
To summarise, we may write
\begin{equation} \label{eqn:V_I_nobalk}
U_I(1-P_J,P_J)=\left(1-\rho^{n_e}\left[1-\frac{1-\rho}{1-P_J \rho} \right]\right)^{-1} \hspace{-10pt} \times \left[R(1-\rho^{n_e})-b\right] - C_I.
\end{equation}

Recall that the equilibrium probability $P_J^*$ satisfies $y^*=1-P_J^*\rho$ where $y^*$ is the one positive root of the quadratic equation in \eqref{eqn:quad_simple_S5}. Using the quadratic formula, we may write
$$1-P_J^*\rho=\frac{(1-\rho)(\sqrt{(R+d)^2+f}-R-d)}{k}$$
where
$$ k=2C_I(\rho^{-n_e}-1)>0, \: \: f=\frac{2 k c_w}{\mu (1-\rho)}>0, \: \: \text{and} \: \: d=C_I-\frac{c_w n_e}{\mu}$$
all do not depend on $R$. Substitution into \eqref{eqn:V_I_nobalk} yields
\begin{equation} \label{eqn:V_I_nobalk_at_equil}
U_I(1-P_J^*,P_J^*)=\left(1-\rho^{n_e}\left[1-\frac{k}{\sqrt{(R+d)^2+f}-R-d} \right]\right)^{-1}  \hspace{-10pt} \times \left[R(1-\rho^{n_e})-b\right] - C_I. 
\end{equation}
To complete the proof, we will show that the derivative of \eqref{eqn:V_I_nobalk_at_equil} with respect to $R$ is strictly positive.

Write $g(R)=\sqrt{(R+d)^2+f}$. Then we have $g'(R)=(R+d)/g(R)$, and hence
$$\frac{d}{dR}(g(R)-R-d)^{-1}=\frac{1-g'(R)}{(g(R)-R-d)^2}=\frac{1}{g(R)(g(R)-R-d)}.$$
We may now use the quotient rule to determine that the derivative of $U_I(1-P_J^*,P_J^*)$ is equal to the following divided by a squared (and hence strictly positive) term:
\begin{align}
&(1-\rho^{n_e})\left(1-\rho^{n_e}\left[1-\frac{k}{g(R)-R-d} \right]\right)-\frac{k\rho^{n_e}(R(1-\rho^{n_e})-b)}{g(R)(g(R)-R-d)} \nonumber \\
=&(1-\rho^{n_e})^2+k\rho^{n_e} \left[\frac{(g(R)-R)(1-\rho^{n_e})+b}{g(R)(g(R)-R-d)} \right] \nonumber \\
=&\frac{k\rho^{n_e}(1-\rho^{n_e})}{g(R)(g(R)-R-d)}\left[\frac{g(R)(g(R)-R-d)}{2C_I}+g(R)-R+\frac{b}{1-\rho^{n_e}}\right].
\label{eqn:proof_deriv}
\end{align}
The proof will be completed by showing that \eqref{eqn:proof_deriv} is strictly positive. Recall that $k>0$, and note that $g(R)>R+d>0$; it follows that the term multiplying the squared brackets in \eqref{eqn:proof_deriv} is strictly positive. Since $b>0$ and $g(R)-R>d$, the bracketed term is strictly greater than
\begin{equation*}
\frac{g(R)(g(R)-R-d)}{2C_I}+d>\frac{\widehat{g}(R)(\widehat{g}(R)-R-d)}{2C_I}+d,
\end{equation*}
where 
$$\widehat{g}(R) \equiv \sqrt{(R+d)^2+4C_Ic_w/\mu} < \sqrt{(R+d)^2+f} \equiv g(R).$$
Therefore, to show that \eqref{eqn:proof_deriv} is strictly positive, it will suffice to show that $\widehat{g}(R)(\widehat{g}(R)-R-d)+2dC_I>0$. We have
\begin{equation} \label{eqn:proof_final_exp}
\widehat{g}(R)(\widehat{g}(R)-R-d)+2dC_I=(R+d)^2+\frac{2C_Ic_wn_e}{\mu}-(R+d)\sqrt{(R+d)^2+4C_Ic_wn_e/\mu} +2C_I^2.
\end{equation}
When $n_e=0$, \eqref{eqn:proof_final_exp} clearly simplifies to $2C_I^2$. The proof is completed by noting that the derivative of \eqref{eqn:proof_final_exp} with respect to $n_e$ is given by
\begin{equation*}
\frac{2C_Ic_w}{\mu}\left(1-\frac{R+d}{\sqrt{(R+d)^2+4C_Ic_wn_e/\mu}}\right)>0.
\end{equation*}
\end{proof}

\subsection{Proof of Lemma \ref{prop:equil_prob_cross_threshold}} \label{append:proof:prop:equil_prob_cross_threshold}
\begin{proof}
Let $x \in \{2,3,\ldots\}$. First suppose that $R=xc_w/\mu$, so the optimal threshold satisfies $n_e=x-1$. By \eqref{eqn:quad_simple_S5}, the equilibrium probability $P_J^*$ satisfies $y^*=1-P_J^*\rho$ where $y^*$ is the one positive root of the quadratic equation
\begin{equation} \label{eqn:feqn_proof}
f(y) \equiv y^2\left(C_I\frac{\rho^{-(x-1)}-1}{1-\rho}\right)+y\left(C_I+\frac{c_w}{\mu}\right)-\frac{c_w}{\mu}=0.
\end{equation}
Now take $\epsilon\in (0,c_w/\mu)$ and set $R=xc_w/\mu+\epsilon$. Now the optimal threshold is $n_e=x$, and, by \eqref{eqn:quad_simple_S5}, the equilibrium probability $P_{J,\epsilon}^*$ now satisfies $y^*_{\epsilon}=1-P_{J,\epsilon}^*\rho$, where $y^*_{\epsilon}$ is the one positive root of
\begin{equation*}
g(y) \equiv y^2\left(C_I\frac{\rho^{-x}-1}{1-\rho}\right)+y\left(C_I+\epsilon\right)-\frac{c_w}{\mu}=0.
\end{equation*}
To complete the proof, we will show that $y^*_{\epsilon}>y^*$ for small-enough $\epsilon$, from which it clearly follows that $P_{J,\epsilon}^*<P_{J}^*$.

Note that, for any $k\in\{1,2,\ldots\}$ we have
$$\frac{\rho^{-k}-1}{1-\rho}=\sum_{i=1}^k \rho^{-i},$$
from which it follows that
\begin{equation} \label{eqn:gf_diff}
g(y)-f(y)=y^2C_I\rho^{-x}-y\left(\frac{c_w}{\mu}-\epsilon\right).
\end{equation}
Therefore, $f$ and $g$ intersect at $y=0$ and at
\begin{equation} \label{eqn:yhat_def}
\widehat{y}\equiv \frac{\rho^x}{C_I}\left(\frac{c_w}{\mu}-\epsilon\right)>0.
\end{equation}
When $|y|$ is very large, the positive quadratic term in \eqref{eqn:gf_diff} dominates, so we must have $g(y)>f(y)$. It follows that we have $g(y) < f(y)$ if and only if $y \in(0,\widehat{y})$. 
Since their quadratic terms are positive and shared constant term negative, $f$ and $g$ have one negative and one positive root, the positive roots denoted by $y^*$ and $y^*_{\epsilon}$, respectively. 
It follows that $f$ must be strictly negative on $[0,y^*)$ and strictly positive on $(y^*, \infty)$, and that $g$ must be strictly negative on $[0,y^*_{\epsilon})$ and strictly positive on $(y^*_{\epsilon}, \infty)$. 

We have $g(\widehat{y})=f(\widehat{y})$ by the definition of $\widehat{y}$. By the deductions in the preceding paragraphs, note that if $g(\widehat{y})=f(\widehat{y})>0$, then we must have $y^*_{\epsilon}, y^* < \widehat{y}$, so $f(y^*)=0 > g(y^*)$, and hence $y^* \in (0,y^*_{\epsilon})$. The proof will therefore be completed by showing that $f(\widehat{y})>0$ for small-enough $\epsilon$.

By the definitions in \eqref{eqn:feqn_proof} and \eqref{eqn:yhat_def}, we have
\begin{equation*}
f(\widehat{y}) = \left(\frac{\rho^x}{C_I}\left(\frac{c_w}{\mu}-\epsilon\right) \right)^2\left(C_I\frac{\rho^{-(x-1)}-1}{1-\rho}\right)+\frac{\rho^x}{C_I}\left(\frac{c_w}{\mu}-\epsilon \right)\left(C_I+\frac{c_w}{\mu}\right)-\frac{c_w}{\mu}.
\end{equation*}
Collecting terms which do not depend on $\epsilon$, we obtain
\begin{align*}
f(\widehat{y})&=\left(\frac{c_w \rho^x}{\mu}\right)^2 \left(\frac{1}{C_I}\right)\left(\frac{\rho^{-(x-1)}-1}{1-\rho}\right)+\left(\frac{c_w \rho^x}{C_I\mu} \right)\left(C_I+\frac{c_w}{\mu}\right) - \frac{c_w}{\mu} + \epsilon(a\epsilon +b) 
\end{align*}
where $a$ and $b$ do not depend on $\epsilon$. Simplifying further, we obtain
\begin{equation} \label{eqn:feqn_athat_proof}
f(\widehat{y})=\frac{c_w(1-\rho^x)}{\mu}\left(\frac{c_w \rho^x}{C_I (\mu-\lambda)} -1 \right)+ \epsilon(a\epsilon +b).
\end{equation}

Since the equilibrium satisfies $P_I>0$, by Table \ref{tab:HR_no_custs}, in scenario 1 we have $C_I < C_{I,0}^1(R)$, and in scenarios 2 and 3 we have $C_I < C_{I,0}^{2,3}(R)$. Recall that $C_{I,0}^1(R)=C_{I,0}^{2,3}(R)$ when $R=c_w/(\mu-\lambda)$. By their properties discussed in Section \ref{sec:no_custs}, we have $C_{I,0}^1(R)$ continuous and strictly decreasing over all scenarios and $C_{I,0}^{2,3}(R)$ continuous and strictly increasing over all scenarios. Therefore, when $R<c_w(\mu-\lambda)$, which covers all of scenarios 2 and 3, we have $C_{I,0}^{2,3}(R)<C_{I,0}^1(R)$. We may hence conclude that in \emph{any} scenario, if the equilibrium has $P_I>0$, then we have $C_I<C_{I,0}^1(R)$.

Using the definition in Table \ref{tab:HR_no_custs} with $n_e=x-1$, we have
$$C_I<C_{I,0}^1(xc_w/\mu)=\frac{\rho^xc_w}{(\mu-\lambda)}.$$
It follows that the left-hand term in \eqref{eqn:feqn_athat_proof} is strictly positive when the equilibrium satisfies $P_I>0$. Therefore, for small-enough $\epsilon>0$, we have $f(\widehat{y})>0$, completing the proof.
\end{proof}

\end{appendices}

\bibliographystyle{plain}
\bibliography{References}

\end{document}